\documentclass[11pt]{article} 

\usepackage[margin=1in]{geometry}
\usepackage{float}
\usepackage{amssymb,amsfonts,amsmath}
\usepackage{enumerate}
\usepackage{enumitem}
\usepackage{caption}
\usepackage{amsthm}
\usepackage{color}
\usepackage[colorlinks=true, linkcolor=red, urlcolor=blue, citecolor=blue]{hyperref}
\usepackage{comment}

\newcommand{\R}{\mathbb{R}}                     
\newcommand{\Z}{\mathbb{Z}}

\def\var#1{\mbox{\bf Var}[ #1 ]}

\newcommand{\poly}{\text{poly}}

\newcommand{\pr}[1]{\text{\normalfont Pr}\normalfont\lbrack #1 \rbrack} 
\newcommand{\ex}[1]{\mathbb{E}\normalfont\lbrack #1 \rbrack}
\newcommand{\bpr}[1]{\text{\normalfont Pr}\normalfont \left[#1 \right]} 
\newcommand{\bex}[1]{\mathbb{E}\normalfont \left[#1 \right]}

\newcommand{\eps}{\epsilon}

\newcommand{\X}{\mathcal{X}} 

\newcommand{\Y}{\mathcal{Y}}

\newcommand{\ttx}[1]{\texttt{#1}}

\newtheorem{theorem}{Theorem}
\newtheorem{lemma}{Lemma}
\newtheorem{corollary}{Corollary}
    \newtheorem{fact}{Fact}
\newtheorem{proposition}{Proposition}
\newtheorem{definition}{Definition}

\newtheorem{remark}{Remark}

\title{Towards Optimal Moment Estimation in Streaming and Distributed Models}
\author{
	Rajesh Jayaram\\
	Carnegie Mellon University\\
	\texttt{rkjayara@cs.cmu.edu}
	\and
	David P. Woodruff\\
	Carnegie Mellon University \\
	\texttt{dwoodruf@cs.cmu.edu\footnote{The authors thank the partial support by the National Science Foundation under Grant No. CCF-1815840.}}
}
\date{}

\begin{document}
	\clearpage\maketitle
	\thispagestyle{empty}
	\setcounter{page}{0}
	\maketitle
	\begin{abstract}
	
	One of the oldest problems in the data stream model is to approximate the $p$-th moment $\|\X\|_p^p = \sum_{i=1}^n \X_i^p$ of an underlying non-negative vector $\X \in \R^n$, which is presented as a sequence of $\poly(n)$ updates to its coordinates. Of particular interest is when $p \in (0,2]$. Although a tight space bound of $\Theta(\epsilon^{-2} \log n)$ bits is known for this problem when both positive and negative updates are allowed, surprisingly there is still a gap in the space complexity of this problem when all updates are positive. Specifically, the upper bound is $O(\epsilon^{-2} \log n)$ bits, while the lower bound is only $\Omega(\epsilon^{-2} + \log n)$ bits. Recently, an upper bound of $\tilde{O}(\epsilon^{-2} + \log n)$ bits was obtained under the assumption that the updates arrive in a {\it random order}.

We show that for $p \in (0, 1]$, the random order assumption is not needed. Namely, we give an upper bound for worst-case streams of $\tilde{O}(\epsilon^{-2} + \log n)$ bits for estimating $\|\X\|_p^p$. Our techniques also give new upper bounds for estimating the empirical entropy in a stream. On the other hand, we show that for $p \in (1,2]$, in the natural coordinator and blackboard distributed communication topologies, there is an  $\tilde{O}(\epsilon^{-2})$ bit max-communication upper bound based on a randomized rounding scheme. 
Our protocols also give rise to protocols for heavy hitters and approximate matrix product. We generalize our results to arbitrary communication topologies $G$, obtaining an $\tilde{O}(\epsilon^{2} \log d)$ max-communication upper bound, where $d$ is the diameter of $G$. Interestingly, our upper bound rules out natural communication complexity-based approaches for proving an  $\Omega(\epsilon^{-2} \log n)$ bit lower bound for $p \in (1,2]$ for streaming algorithms. In particular, any such lower bound must come from a topology with large diameter.

	\end{abstract}
	\newpage

	\section{Introduction}
		
	The streaming and distributed models of computation have become increasingly important for the analysis of massive datasets, where the sheer size of the input imposes stringent restrictions on the resources available to algorithms. 
Examples of such datasets include internet traffic logs, sensor networks, financial transaction data, database logs, and scientific data streams (such as huge experiments in particle physics, genomics, and astronomy).
Given their prevalence, there is a large body of literature devoted to designing extremely efficient algorithms for analyzing streams and enormous datasets. We refer the reader to \cite{babcock2002models, muthukrishnan2005data} for surveys of these algorithms and their applications.  

Formally, the data stream model studies the evolution of a vector $\X \in \Z^n$, called the frequency vector\footnote{Our motivation for using the notation $\X$, as opposed to the more common $x$ or $f$, is to align the streaming notation with the multi-party communication notation that will be used throughout the paper.}. Initially, $\X$ is initialized to be the zero-vector. The frequency vector then receives a stream of $m$ coordinate-wise updates of the form $(i_t,\Delta_t) \in [n] \times \{-M,\dots,M\}$ for some $M > 0$ and time step $t \in [m]$. Each update $(i_t, \Delta_t)$ causes the change $\X_{i_t} \leftarrow \X_{i_t} + \Delta_{t}$. If we restrict that $\Delta_t \geq 0$ for all $t \in [m]$, this is known as the \textit{insertion-only} model. If the updates $\Delta_t \in \{-M,\dots,M\}$  can be both positive and negative, then this is known as the \textit{turnstile}-model. The $p$-th frequency moment of the frequency vector at the end of the stream, $F_p$, is defined as $F_p = \sum_{i=1}^n |\X_i|^p$. For simplicity (but not necessity), it is generally assumed that $m,M = \poly(n)$. 

The study of frequency moments in the streaming model was initiated by the seminal 1996 paper of Alon, Matias, and Szegedy \cite{alon1996space}. Since then, nearly two decades of research have been devoted
to understanding the space and time complexity of this problem. An incomplete list of works which study frequency moments in data streams includes \cite{chakrabarti2003near,indyk2005optimal,bar2004information, woodruff2004optimal,indyk2006stable,kane2010exact,braverman2010recursive,kane2011fast,braverman2014optimal,chakrabarti2016robust,braverman2014optimal,blasiok2017continuous, braverman2018revisiting}. For $p>2$, it is known that polynomial in $n$ (rather than logarithmic) space is required for $F_p$ estimation \cite{chakrabarti2003near,indyk2005optimal}.  In the regime of $p \in (0,2]$, the space complexity of $F_p$ estimation in the turnstile model is now understood, with matching upper and lower bounds of $\Theta(\eps^{-2}\log(n))$ bits to  obtain a $(1 \pm \eps)$ approximation of $F_p$. Here, for $\eps > 0$, a $(1 \pm \eps)$ approximation means an estimate $\tilde{F_p}$ such that $(1-\eps)F_p \leq \tilde{F_p}\leq (1+\eps)F_p$. For insertion only streams, however, the best known lower bound is $\Omega(\eps^{-2} + \log(n))$ \cite{woodruff2004optimal}. Moreover, if the algorithm is given query access to an arbitrarily long string of random bits (known as the random oracle model), then the lower bound is only $\Omega(\eps^{-2})$. On the other hand, the best upper bound is to just run the turnstile $O(\eps^{-2}\log(n))$-space algorithm. 

In this work, we make progress towards resolving this fundamental problem. For $p<1$, we resolve the space complexity by giving an $\tilde{O}(\eps^{-2} + \log n)$\footnote{the $\tilde{O}$ here suppresses a single $(\log \log n + \log 1/\eps)$ factor, and in general we use $\tilde{O}$ and $\tilde{\Omega}$ to hide $\log \log n$ and $\log 1/\eps$ terms.}-bits of space upper bound. In the random oracle model, our upper bound is $\tilde{O}(\eps^{-2})$\footnote{This space complexity is measured \textit{between updates}. To read and process the $\Theta(\log(n))$-bit identity of an update, the algorithm will use an additional $O(\log(n))$-bit working memory tape during an update. Note that all lower bounds only apply to the space complexity between updates, and allow arbitrary space to process updates.}, which also matches the lower bound in this setting. Prior to this work, an  $\tilde{O}(\eps^{-2} + \log(n))$ upper bound for $F_p$ estimation was only known in the restricted \textit{random-order} model, where it is assumed that the stream updates are in a uniformly random ordering \cite{braverman2018revisiting}. Our techniques are based on novel analysis of the behavior of the $p$-stable random variables used in the $O(\eps^{-2}\log(n))$ upper bound of \cite{indyk2006stable}, and also give rise to a space optimal algorithm for entropy estimation.

We remark that $F_p$ estimation in the range $p \in (0,1)$ is useful for several reasons. Firstly, for $p$ near $1$, $F_p$ estimation is often used as a subroutine for estimating the empirical entropy of a stream, which itself is useful for network anomaly detection (\cite{li2011new}, also see  \cite{harvey2008sketching} and the references therein). Moment estimation is also used in weighted sampling algorithms for data streams \cite{monemizadeh20101,Jowhari:2011, jayaram2018perfect} (see \cite{cormode2019p} for a survey of such samplers and their applications). Here, the goal is to sample an index $i \in [n]$ with probability $|\X_i|^p/F_p$. These samplers can be used to find heavy-hitters in the stream, estimate cascaded norms \cite{andoni2010streaming, monemizadeh20101}, and design representative histograms of $\X$ on which more complicated algorithms are run  \cite{gibbons1997fast,  gibbons1998new,olken1993random,gilbert2002summarize,huang2007communication,cormode2005summarizing}. Furthermore, moment estimation for fractional $p$, such as $p = .5$ and $p=.25$, has been shown to be useful for data mining 
\cite{cormode2002fast}. 

For the range of $p \in (1,2]$, we prove an $\tilde{O}(\eps^{-2})$-bits of max-communication upper bound in the distributed models most frequently used to prove \textit{lower bounds} for streaming. This result rules out a large and very commonly used class of approaches for proving lower bounds against the space complexity of streaming algorithms for $F_p$ estimation. Our approach is based on a randomized rounding scheme for $p$-stable sketches. Along the way, we will prove some useful inequalities to bound the heavy-tailed error incurred from rounding non-i.i.d. $p$-stable variables for $p<2$. We show that our rounding scheme can be additionally applied to design improved protocols for the distributed heavy hitters and approximate matrix product problems. We now introduce the model in which all the aforementioned results hold. 

	\paragraph{Multi-Party Communication}
In this work, we study a more general model than streaming, known as the message passing multi-party communication model. All of our upper bounds apply to this model, and our streaming algorithms are just the result of special cases of our communication protocols. In the message passing model, there are $m$ players, each positioned at a unique vertex in a graph $G = (V,E)$. The $i$-th player is given as input an integer vector $X_i \in \Z^n$. The goal of the players is to work together to jointly approximate some function $f:\R^n \to \R$ of the aggregate vector $\X = \sum_{i=1}^n X_i$, such as the $p$-th moment $f(\X)=F_p =  \|\X\|_p^p=\sum_{i=1}^n |\X_i|^p$. In the message passing model, as opposed to the \textit{broadcast} model of communication, the players are only allowed to communicate with each other over the edges of $G$. Thus player $i$ can send a message to player $j$ only if $(i,j) \in E$, and this message will only be received by player $j$ (and no other). At the end of the protocol, it is assumed that at least one player holds the approximation to $f(\X)$. The goal of multi-party communication is to solve the approximation problem using small total communication between all the players over the course of the execution. More specifically, the goal is to design protocols that use small \textit{max-communication}, which is the total number of bits sent over any edge of $G$. Our protocols hold in an even more restricted setting, known as the \textit{one-shot} setting, where each player is allowed to communicate exactly once over the course of the entire protocol.

The message passing setting is often used to model distributed computation, where there are $m$ distributed machines, arranged in some network topology $G$, and their goal is to jointly compute some function of the aggregate data $\X$. Oftentimes, the predominant bottleneck in distributed data processing is the network bandwidth and the energy cost of communication \cite{juang2002energy,madden2005tinydb}, and so it is desirable to have algorithms that use as little communication over any edge as possible.

We now observe that data streams can be modeled as a special case of one-shot multi-party communication. Here, the graph $G$ in question is the line graph on $m$ vertices. If the updates to the data stream vector are $(i_1,\Delta_1),\dots,(i_m,\Delta_m),$ then the $t$-th player has input $X_t \in \Z^n$, where $(X_t)_{i_t} = \Delta_t$ and $(X_t)_j = 0$ for $j \neq i_t$. The aggregate vector $\X = \sum_{i=1}^m X_i$ is just the frequency vector at the end of the stream, and the space complexity of any algorithm is just the max-communication used over any edge of the corresponding communication protocol. 
Since we are primarily interested in insertion only streams, in this work we will consider the \textit{non-negative data} model, where $X_i \in \{0,1,\dots,M\}^n$ for all input vectors $X_i$, for some $M >0$ (as in streaming, we assume $M = \poly(n,m)$ for simplicity). Note that an equivalent condition is that each $X_i \in \R^n_{\geq 0}$ such that the entries of $X_i$ can be stored in $O(\log M)$-bits.

We are now ready to introduce our results for moment estimation in the message passing model. Let $d$ be the \textit{diameter} of the communication graph $G$. Our first result is a protocol for $F_p$ estimation when $p \in (1,2]$ which uses a max communication of $\tilde{O}(\eps^{-2} \log d)$ bits. Using similar techniques, we also obtain a (optimal for $d = \Theta(1)$) bound of $\tilde{O}(\eps^{-2} \log n \log d)$ for the heavy hitters problem, which is to find the coordinates of $\X$ which contribute at least an $\eps$ fraction of the total $\sqrt{F_2} = \|\X\|_2$ of $\X$. For $p \in (0,1)$, we give an  $\tilde{O}(\eps^{-2})$ upper bound for $F_p$ estimation. Notice that this is independent of the graph topology, and thus holds for the line graph, where we derive our $\tilde{O}(\eps^{-2})$ upper bound for $F_p$ estimation in the random oracle streaming model. We then show how the streaming algorithm can be derandomized to not require a random oracle, now using an optimal $\tilde{O}(\eps^{-2} + \log(n))$-bits of space. Our techniques also result in an $\tilde{O}(\eps^{-2})$ upper bound for additively approximating the empirical entropy of the vector $\X$. 

Our results for $p \in (1,2]$ have interesting implications for any attempts to prove \textit{lower-bounds} for streaming algorithms that estimate $F_p$, which we now describe. 
The link between streaming and communication complexity is perhaps one of the most fruitful sources of space lower bounds for algorithms in computer science. Namely, nearly all lower bounds for the space complexity of randomized streaming algorithms are derived via reductions from communication problems. For an incomplete list of such reductions, see \cite{woodruff2004optimal,woodruff2018distributed,kane2010exact, Jowhari:2011, kapralov2017optimal, braverman2016streaming,chakrabarti2003near, weinstein2015simultaneous,li2013tight, mcgregor2016space,jayram2009data} and the references therein. Now nearly all such lower bounds (and all of the ones that were just cited) hold in either the $2$-party setting ($G$ has $2$ vertices), the coordinator model, or the black-board model. In the coordinator model there are $m$ players, each with a single edge to a central coordinator (i.e., $G$ is a star graph on $m+1$ vertices). Note that the diameter $d$ of the coordinator graph is $2$. In the multi-player black-board model, every message that is sent is written to a shared blackboard that can be read by all players. Observe that any one-way protocol for the coordinator model immediately results in a protocol with the same communication for the blackboard model. Namely, each player simply writes what it would have sent to the coordinator on the blackboard, and at the end of the protocol the blackboard contains all the information that the coordinator would have had. For these three settings, our protocol gives an $\tilde{O}(\eps^{-2})$ max-communication upper bound for $F_p$ estimation, $p \in (1,2]$. This completely rules out the approach for proving lower bounds against $F_p$ estimation in a stream via any of these three techniques. In particular, it appears that any lower bound for $F_p$ estimation via communication complexity in this regime of $p$ will need to use a graph with $\Omega(n)$ diameter, such as the line graph, without a black-board. 

The coordinator and black-board models have also been studied in many other settings than for proving lower bounds against streaming. For instance, in the \textit{Distributed Functional Monitoring} literature \cite{cormode2011algorithms,yi2013optimal,woodruff2012tight,huang2012randomized,tirthapura2011optimal,JayaramWeighted:2018}, each player is receiving a continuous stream of updates to their inputs $X_i$, and the coordinator must continuously update its approximation to $f(\X)$. The black-board model is also considered frequently for designing communication upper bounds, such as those for set disjointness \cite{bar2004information,chakrabarti2003near,gronemeier2009asymptotically}. Finally, there is  substantial literature which considers numerical linear algebra and clustering problems in the coordinator model \cite{woodruff2018distributed,chen2016communication,balcan2016communication,woodruff2016distributed}. Thus, our upper bounds can be seen as a new and useful contribution to these bodies of literature as well.

\paragraph{Numerical Linear Algebra} The study of numerical linear algebra in the distributed model has become increasingly important  \cite{feldman2013turning, liang2014improved,liberty2013simple,ghashami2014relative,kannan2014principal,bhojanapalli2015tighter,  woodruff2018distributed,boutsidis2016optimal}, especially as frameworks like Hadoop \cite{hadoop} and Spark \cite{spark} become critical for analyzing massive datasets. These works 
prioritize designing protocols to compute linear algebraic primitives of matrices distributed over multiple servers, using as little communication as possible. Motivated by this, we apply the techniques developed for our $F_p$ estimation problems to the linear algebraic primitive of approximate matrix product.

The setting is as follows. Instead of vector-valued inputs $X_i \in \Z^n_{\geq 0}$, the players are given $X_i \in \Z^{n \times t_1}_{\geq 0 }$, and they would like to compute statistics of the aggregate \textit{data matrix} $\X = \sum_{i=1}^m X_i$.  We remark that in the domain of application, we generally  assume $n >> t_1,t_2$ (however this is not, strictly speaking, a requirement). So, for instance, they may want to estimate the $t_1 \times t_2$ dimensional product of $\X^T$ with another matrix $\Y \in \R^{n \times t_2}$. In this case, each player also gets as input a $Y_i \in \Z^{n \times t_2}$, and we set $\Y = \sum_{i=1}^m Y_i$. The goal of the approximate matrix product problem is for the players to jointly compute an approximation to $\X^T\Y$. Specifically, they would like to obtain a $R \in \R^{t_1 \times t_2}$ such that $\|R - \X^T \Y\|_F \leq \eps \|\X\|_F \|\Y\|_F$, where the Frobenius norm $\|\X\|_F$ of a matrix $\X$ is just the square root of the sum of squares of entries in $\X$.  We note that a special and important case of interest is covariance estimation, where $X_i = Y_i$. Here, each player $i$ has some subset of a positive data matrix $\X$, and the goal is to compute the empirical covariance $\X^T \X$ of the data.

 Using similar techniques as our $F_p$ estimation algorithm for $p>1$, we design an $\tilde{O}((t_1 + t_2)\eps^{-2} \log d)$-bit max communication protocol for the problem. We note here that a $O((t_1 + t_2)\eps^{-2}\log n) $ upper bound is a standard and well-known result from sketching, and our main contribution here is to reduce the $\log n$ to a $\log d$. Since the majority of works in the distributed linear algebra literature consider the coordinator model, our results demonstrate an improvement of a $\log(n)$ factor in the communication for approximate matrix product in this setting.

\paragraph{Linear Sketches} Our results are closely related to the theory of linear sketching, which is the standard approach used to solve many problems which involve approximating functions of extremely high dimensional data $\X$. Formally, a linear sketch is a randomized matrix $S \in \R^{k \times n}$, where $k <<n$, such that given $S \X \in \R^k$, one can approximate the desired function of $\X$. To estimate $F_p$, $p \in (0,2]$, such a sketch $S$ exists with $k = \Theta(\eps^{-2})$ rows (\cite{indyk2006stable, kane2010exact}). If $p=2$, then $S$ is a matrix of i.i.d. Gaussian random variables, and for general $p \in (0,2]$, so called $p$-stable distributions, $D_p$, are used. These distributions have the property that if $Z_1,\dots,Z_m \sim D_p$, then $\sum_{i=1}^m Z_i \X_i \sim Z \|\X\|_p$, where $Z \sim D_p$ again. 

To solve the multi-party $F_p$ estimation problem using sketches, each player can locally compute $SX_i$ and pass it along to some centralized vertex, who can then compute $S\X = \sum_i SX_i$. Given $S\X$, the central vertex can then compute the desired approximation. If the players sum up their sketches along the way, as they travel up to the central vertex, the per-player communication of this protocol is $O(\eps^{-2}\log(n))$, where the $\log(n)$ factor is required to store each coordinate of $S\X$. As noted, however, the extra $\log(n)$ factor does not match the $\Omega(\eps^{-2})$ lower bound. Thus, the main challenge for each of the protocols in this paper will be to avoid sending the entire $O(\log n)$-bit counter needed to store a coordinate of the sketch $SX_i$.



	\subsection{Our Contributions}

	As noted, the upper bounds in this paper all hold in the general multi-party message passing model, over an arbitrary topology $G$. Our algorithms also have the additional property that they are \textit{one-shot}, meaning that each player is allowed to communicate exactly once. Our protocols pre-specify a central vertex $\mathcal{C} \in V$ of $G$. Specifically, $\mathcal{C}$ will be a \textit{center} of $G$, which is a vertex with minimal max-distance to any other vertex. Our protocols then proceed in $d$ rounds, where $d$ is the diameter of $G$. Upon termination of the protocols, the central vertex $\mathcal{C}$ will hold the estimate of the protocol. We note that $\mathcal{C}$ can be replaced by any other vertex $v$, and $d$ will then be replaced by the max distance of any other vertex to $v$. A summary of our results is given in Figure \ref{fig:results}.


	\begin{figure}[t]
		\begin{center}
			\begin{tabular}{| c | c| c| c |} 
				\hline
				Problem & Prior best upper bound & Upper bound (this work)   & Notes \\
				\hline
				$F_p$, $1 < p \leq 2$ &$O(\eps^{-2} \log(n))$ \cite{ kane2010exact}& $\tilde{O}(\eps^{-2} \log(d))$ &  \\
			$F_p$, $p < 1$&$O(\eps^{-2} \log(n))$\cite{ kane2010exact}  &$\tilde{O}(\eps^{-2})$ &\\
			$F_p$ Streaming, $p<1$& $O(\eps^{-2} \log(n))$\cite{ kane2010exact} & $\tilde{O}(\eps^{-2})$ & \\
				
				Entropy & --  &$ \tilde{O}(\eps^{-2})$ &   \\
					Entropy  Streaming &$O(\eps^{-2} \log^2(n))$ \cite{clifford2013simple}&$ \tilde{O}(\eps^{-2})$ & random oracle \\ 
					Point Estimation & $O(\eps^{-2} \log^2(n))$ \cite{charikar2002finding} &$\tilde{O}(\eps^{-2}\log(d) \log(n))$ & \\
					Approx Matrix Prod. & -- & $\tilde{O}(1)$ & \begin{tabular}{c}
					
				per coordinate \\ of sketch \\
					\end{tabular} \\ \hline
			\end{tabular} 
		\end{center} \caption{For the communication problems above, the bounds are for the max-communication (in bits) across any edge. For the streaming problems, the bounds are for the space requirements of the algorithm. Here, $d$ is the diameter of the communication network $G$. 
		For all problems except point estimation, there is a matching $\Omega(\eps^{-2})$ lower bound. The problem of point estimation itself has a matching $\Omega(\eps^{-2} \log n)$ lower bound for graphs with constant $d$. }  \label{fig:results}
	\end{figure}

We first formally state our general result for $F_p$ estimation, $1 < p \leq 2$. Note that, while we state all our results for constant probability of success, by repeating $\log(1/\delta)$ times and taking the median of the estimates, this is boosted to $1-\delta$ in the standard way.  \\

\noindent \textbf{Theorem \ref{thm:LPmain}}  \textit{ 
	For $p \in (1,2]$, there is a protocol for $(1 \pm \eps)$ approximating $F_p$ which succeeds with probability $3/4$ in the message passing model. The protocol uses a max communication of $O(\frac{1}{\eps^2}(\log \log n + \log d  + \log 1/\eps ))$ bits, where $d$ is the diameter of $G$.}\\

For graphs with constant diameter, such as the coordinator model, our max communication bound of $\tilde{O}(\eps^{-2})$ matches the $\Omega(\eps^{-2})$ lower bound \cite{woodruff2004optimal,chakrabarti2012optimal}, which follows from a $2$-player reduction from the Gap-Hamming Distance problem. 
For $p=2$, our \textit{total communication} in the coordinator model matches the $\Omega(m^{p-1}/\eps^2)$ total communication lower bound (up to $\log \log(n)$ and $\log(1/\eps)$ terms) for non-one shot protocols \cite{woodruff2012tight}.
For one shot protocols, we remark that there is an $\Omega(m/\eps^2)$ total communication lower bound for any $p \in (0,2] \setminus \{1\}$ (see Appendix \ref{app:1}). 
As discussed previously, our result also has strong implications for streaming algorithms, demonstrating that no $\Omega(\eps^{-2} \log n)$ lower bound for $F_p$ estimation, $p \in (1,2]$, can be derived via the common settings of $2$-party, coordinator, or blackboard communication complexity. 

Our main technique used to obtain Theorem \ref{thm:LPmain} is a new randomized rounding scheme for $p$-stable sketches. Suppose we are in the coordinator model, $Z=(Z_1,\dots,Z_n)$ are i.i.d. $p$-stable, and the players want to jointly compute $\langle Z, \X \rangle = \sum_{i=1}^m \langle Z, X_i\rangle$. Then, roughly speaking, they could each round their values $\langle Z, X_i\rangle$ to $(1 \pm \gamma)$ relative error, using $O( \log \log n + \log 1/\gamma)$ bits of communication, and send it to the coordinator. The error for each player would be at most $\gamma | \langle Z, X_i\rangle|$, for a total error of  $ \gamma\sum_{i=1}^m  | \langle Z, X_i\rangle|$. For $p>1$, however, the final counter value $\langle Z, \X \rangle$ will be much smaller than $\gamma\sum_{i=1}^m  | \langle Z, X_i\rangle|$ -- in fact it will by polynomial smaller. To see this, note that $\langle Z, \X \rangle \sim z \|\X\|_p$, where $z$ is $p$-stable, thus  $\langle Z, \X \rangle = O(\|\X\|_p)$ with good probability. On the other hand, each $ | \langle Z, X_i\rangle|$ will be $\Omega(\|X_i\|_p)$ with good probability, and so $ \sum_{i=1}^m | \langle Z, X_i\rangle|$ can be as large as $\Omega( m^{1-1/p } \|\X\|_p)$ (and in fact larger due to the heavy tails), so the error is too great.  
 
Our solution is to randomly round the values $\langle Z, X_i\rangle$ to relative error $(1 \pm \gamma)$ so that the error is zero mean. The total error can then be bounded via the variance. This is proportional to $\sum_{i=1}^m  | \langle Z, X_i\rangle|^2$, which behaves like $\sum_{i=1}^m z_i^2 \|X_i\|_p^2$, where the $z_i$'s are non-i.i.d. $p$-stable. 
Note that if each $z_i^2$ was constant, by the positivity of the vectors $X_i$, we have  $\sum_{i=1}^m z_i^2 \|X_i\|_p^2 = O( \|\X\|_p^2)$ for any $1 \leq p \leq 2$. However, since $p$-stables have heavy tails, many of the $z_i^2$'s will be super-constant, and in fact many will be polynomially sized in $m$. By a careful analysis of the tail behavior of this sum, we are able to still bound the variance by (roughly) $\gamma\|\X\|_p^2$, where $\gamma$ is the precision to which the randomized rounding is carried out, which will give us our protocol for the coordinator model. Our general protocol is the result of iterative applying this rounding and merging scheme until all sketches reach a central vertex. By a judicious analysis of how the variance propagates, we can demonstrate that $\gamma$ need only be $O(1/d)$, where $d$ is the diameter of $G$, to sufficiently bound the overall error of the protocol.

We then show that this randomized rounding protocol can be applied to give improved communication upper bounds for the \textit{point-estimation} problem. Here, the goal is to output a vector $\tilde{X} \in \R^n$ that approximates $\X$ well coordinate-wise. The result is formally given below in Theorem \ref{thm:HHmain}. \\

\noindent \textbf{Theorem \ref{thm:HHmain}}\textit{
	Consider a message passing topology $G = (V,E)$ with diameter $d$, where the $i$-th player is given as input $X^i \in \Z^n_{\geq 0}$ and $\X = \sum_{i=1}^m X^i$.  Then there is a communication protocol which outputs an estimate $\tilde{\X} \in \R^n$ of $\X$ such that 
	\[ \|\tilde{\X} - \X\|_\infty \leq \eps \|\X_{\text{tail}(\eps^{-2})}\|_2 \]
	with probability $1-1/n^c$ for any constant $c \geq 1$. Here $\X_{\text{tail}(\eps^{-2})}$ is $\X$ with the $\eps^{-2}$ largest (in absolute value) coordinates set equal to $0$. The protocol uses a max communication of $O(\frac{1}{\eps^2}\log(n)(\log \log n + \log d  + \log 1/\eps))$.}\\

For graphs with small diameter, our protocols demonstrate an improvement over the previously best known sketching algorithms, which use space $O(\eps^{-2} \log^2(n))$ to solve the point estimation problem \cite{charikar2002finding}. Note that there is an $\Omega(\eps^{-2} \log n)$-max communication lower bound for the problem. This follows from the fact that point-estimation also solves the $L_2$ \textit{heavy-hitters} problem. Here the goal is to output a set $S \subset [n]$ of size at most $|S| = O(\eps^{-2})$ which contains all $i \in [n]$ with $|\X_i| \geq \eps \|\X\|_2$ (such coordinates are called heavy hitters). The lower bound for heavy hitters is simply the result of the space required to store the $\log(n)$-bit identities of all possible $\eps^{-2}$ heavy hitters. Note that for the heavy hitters problem alone, there is an optimal streaming $O(\eps^{-2} \log(n))$-bits of space upper bound called BPTree \cite{braverman2016bptree}. However, BPTree cannot be used in the general distributed setting, since it crucially relies on the sequential natural of a stream.  

Next, we demonstrate that $F_p$ estimation for $p<1$ is in fact possible with max communication independent of the graph topology. After derandomizing our protocol, this results in a optimal streaming algorithm for $F_p$ estimation, $p<1$, which closes a long line of research on the problem for this particular range of $p$ \cite{woodruff2004optimal,indyk2006stable,kane2010exact,kane2011fast,chakrabarti2016robust, braverman2018revisiting}. \\

\noindent \textbf{Theorem \ref{thm:morrismain}}  \textit{
	For $p \in (0,1)$, there is a protocol for $F_p$ estimation in the message passing model which succeeds with probability $2/3$ and has max-communication of $O(\frac{1}{\eps^2}(\log\log n + \log 1/\eps ))$. }\\

\noindent \textbf{Theorem \ref{thm:streamingMain}}  \textit{
	There is a streaming algorithm for $F_p$ estimation, $p \in (0,1)$, which outputs a value $\tilde{R}$ such that with probability at least $2/3$, we have that $|\tilde{R} - \|X\|_p| \leq \eps \|X\|_p$. The algorithm uses $O((\frac{1}{\eps^2}(\log\log n + \log 1/\eps ) + \frac{\log 1/\eps}{\log\log 1/\eps }\log n  )$-bits of space. In the random oracle model, the space is $O(\frac{1}{\eps^2}(\log\log n + \log 1/\eps ))$. }\\
	
	The above bound matches the $\Omega(\eps^{-2})$ max communication lower bound of \cite{woodruff2004optimal} in the shared randomness model, which comes from $2$-party communication complexity. Moreover, our streaming algorithm matches the $\Omega(\log n)$ lower bound for streaming when a random oracle is not allowed. 
	Our results are derived from observations about the behavior of $p$-stable distributions for $p<1$, followed by a careful analysis of their tail-behavior.  Namely, if $Z_1,\dots,Z_n \sim D_p$ are $p$-stable for $p\leq1$, and if $X_1,\dots,X_m \in \R^n$ are non-negative, then $\sum_{i=1}^m |Z_i X_i|$ is roughly on the same order as $\sum_{i=1}^m Z_i X_i$, which allows us to approximate the latter in small space via approximate counters. 
	
	We also point out an alternative approach for deriving the upper bound of Theorem \ref{thm:morrismain}, which is to use \textit{maximally-skewed $p$-stable distributions} \cite{nolan2009stable, li2009compressed, li2011new}. These distributions have the property for that they are always positive (or negative) for $p<1$, and one could then follow a similar analysis to that provided in this work. We chose instead to consider general $p$-stable variables, and not specifically maximally skewed variables, for several reasons. Firstly, for the purposes of streaming, our derandomization utilizes ideas from \cite{kane2010exact}, which only consider un-skewed symmetric stable variables. Thus, to utilize maximally skewed distributions, one would need to adapt the arguments of \cite{kane2010exact}. Moreover, our arguments give general results for the behavior of $p$-stables, and therefore also apply to the maximally-skewed case. To the best of the authors' knowledge, neither approach was known prior to this work.

As another application of our protocol for $F_p$ estimation, $p < 1$, we demonstrate a communication optimal protocol for additive approximation of the empirical \textit{Shannon entropy} $H(\X)$ of the aggregate vector $\X$. Here, $H = H(\X)$ is defined by $H = \sum_{i=1}^n p_i \log(1/p_i)$ where $p_i = |\X_i|/\|\X\|_1$ for $i \in [n]$. The goal of our protocols is to produce an estimate $\tilde{H} \in \R$ of $H$ such that $|\tilde{H} - H| \leq \eps$. Our result is as follows. \\

\noindent \textbf{Theorem \ref{thm:entropymain}}\textit{
	There is a multi-party communication protocol in the message passing model that outputs a $\eps$-additive error of the Shannon entropy $H$. The protocol uses a max-communication of $O(\frac{1}{\eps^2}(\log\log(n) + \log(1/\eps))$-bits. }\\

Note that for a \textit{multiplicative} approximation of the Shannon entropy, there is a 	 $\tilde{\Omega}(\eps^{-2})$ lower bound \cite{chakrabarti2010near}. For additive estimation, \cite{kane2014sparser} gives a $\Omega(\eps^{-2} \log(n))$ lower bound in the turnstile model. Using a similar reduction, we prove a matching $\Omega(\eps^{-2})$ lower bound for additive $\eps$ approximation in the insertion only model (see Appendix \ref{app:2} for the proof). Furthermore, our protocol directly results in an  $\tilde{O}(\eps^{-2})$-bits of space, insertion only  \textit{streaming} algorithm for entropy estimation in the random oracle model. Here, the random oracle model means that the algorithm is given query access to an arbitrarily long string of random bits. We note that many lower bounds in communication complexity (and all of the bounds discussed in this paper except for the $\Omega(\log n)$ term in the lower bound for $F_p$ estimation) also apply to the random oracle model. Previously, the best known algorithm for the insertion only random oracle model used $O(\eps^{-2} \log(n))$-bits \cite{li2011new,clifford2013simple}, whereas the best known algorithm for the non-random oracle model uses $O(\eps^{-2} \log^2(n))$-bits (the extra factor of $\log(n)$ comes from a standard application of Nisan's pseudo-random generator \cite{nisan1992pseudorandom}). \\

\noindent \textbf{Theorem \ref{thm:entropystream}}
There is a streaming algorithm for $\eps$-additive approximation of the empirical Shannon entropy of an insertion only stream in the random oracle model, which succeeds with probability $3/4$. The space required by the algorithm is $O(\frac{1}{\eps^2}(\log\log(n) + \log(1/\eps))$ bits.\\

Finally, we show how our techniques can be applied to the important numerical linear algebraic primitive of \textit{approximate matrix product}, which we now define. 

\begin{definition}
	The multi-party approximate matrix product problem is defined as follows.  Instead of vector valued inputs, each player is given $X_i \in \{0,1,\dots,M\}^{n \times t_1}$ and $Y_i\in \{0,1,\dots,M\}^{n \times t_2}$, where $\X = \sum_i X_i$ and $\Y = \sum_i Y_i$. Here, it is generally assumed that $n>>t_1,t_2$ (but not required). The players must work together to jointly compute a matrix $R \in \R^{t_1 \times t_2}$ such that 
	\[	\|R - \X^T \Y \|_F \leq \eps \|\X\|_F\|\Y\|_F	\]
	Where for a matrix $A \in \R^{n \times m}$, $\|A\|_F = (\sum_{i=1}^n \sum_{j=1}^m A_{i,j}^2)^{1/2}$ is the Frobenius norm of $A$.
\end{definition}

\noindent \textbf{Theorem \ref{thm:approxmatrix}}\textit{
 There is a protocol which outputs, at the central vertex $\mathcal{C}$, a matrix $R \in \R^{t_1 \times t_2}$ which solves the approximate communication protocol with probability $3/4$ \footnote{We remark that there are standard techniques to boost the probability of the matrix sketching results to $1-\delta$, using a blow-up of $\log(\delta)$ in the communication. See e.g. Section 2.3 of \cite{woodruff2014sketching}}.
	The max communication required by the protocol is $O\left(\eps^{-2}(t_1 + t_2)( \log \log n + \log 1/\eps + \log d	)\right)$, where $d$ is the diameter of the communication topology $G$.
}\\

We remark that an upper bound of $O\left(\eps^{-2}(t_1 + t_2) \log n\right)$ was already well-known from sketching theory \cite{woodruff2014sketching}, and our main improvement is removing the $\log(n)$ factor for small diameter graphs, such as the coordinator model where distributed numerical linear algebra is usually considered.

	\subsection{Other Related Work}
	As mentioned, a closely related line of work is in the \textit{distributed functional monitoring model}. Here, there are $m$ machines connected to a central coordinator (the coordinator topology). Each machine then receives a stream of updates, and the coordinator must maintain at all time steps an approximation of some function, such as a moment estimation or a uniform sample, of the union of all streams. We note that there are two slightly different models here. One model is where the items (coordinates) being updated in the separate streams are considered disjoint, and each time an insertion is seen it is to a unique item. This model is considered especially for the problem of maintaining a uniform sample of the items in the streams \cite{cormode2011algorithms, huang2012randomized, tirthapura2011optimal,JayaramWeighted:2018}. The other model, which is more related to ours, is where each player is receiving a stream of updates to a \textit{shared} overall data vector $\X \in \R^n$. This can be seen as a distributed streaming setting, where the updates to a centralized stream are split over $m$ servers, and is considered in \cite{woodruff2012tight,cormode2011algorithms,arackaparambil2009functional}. For the restricted setting of \textit{one-way} algorithms, which only transmit messages from the sites to the coordinators, any such algorithm can be made into a one-shot protocol for the multi-party message passing model. Here, each machine just simulates a stream on their fixed input vectors $X_i$, and sends all the messages that would have been sent by the functional monitoring protocol.
	
Perhaps the most directly related result to our upper bound for for $F_p$ estimation, $p \in (1,2]$, is in the distributed functional monitoring model, where Woodruff and Zhang \cite{woodruff2012tight} show a $O(m^{p-1}\poly(\log(n),1/\eps)  + m\eps^{-1}\log(n) \log(\log(n)/\eps))$\footnote{We remark that the $\poly(\log(n),1/\eps)$ terms here are rather large, and not specified in the analysis of \cite{woodruff2012tight}.} \textit{total communication} upper bound.
We remark here, however, that the result of \cite{woodruff2012tight} is incomparable to ours for several reasons. Firstly, their bounds are only for total communication, whereas their max communication can be substantially larger than $O(1/\eps^2)$. Secondly, while it is claimed in the introduction that the protocols are one way (i.e., only the players speak to the coordinator, and not vice versa), this is for their threshold problem and not for $F_p$ estimation\footnote{The reason for this is as follows. Their algorithm reduces $F_p$ estimation to the threshold problem, where for a threshold $\tau$, the coordinator outputs $1$ when the $F_p$ first exceeds $\tau (1+\eps)$, and outputs $0$ whenever the $F_p$ is below $\tau (1-\eps)$. To solve $F_p$ estimation, one then runs this threshold procedure for the $\log(mMn)/\eps$ thresholds $\tau = (1+\eps), (1+\eps)^2, \dots, (mMn)^2$ in parallel. However, the analysis from \cite{woodruff2012tight} only demonstrates a total communication of $O(k^{1-p} \poly(\log(n),\eps^{-1}))$ for the time steps \textit{before} the threshold $\tau$ is reached. Once the threshold is reached, the communication would increase significantly, thus the coordinator must inform all players when a threshold $\tau$ is reached so that they stop sending messages for $\tau$, violating the one-way property. This step also requires an additive $k$ messages for each of the $O(\eps^{-1}\log(n))$ thresholds, which results in the $O(m\eps^{-1}\log(n) \log(\log(n)\eps)))$ term.}. As remarked before, there is an $\Omega(m/\eps^2)$ total communication lower bound for one-way protocols, which demonstrates that their complexity could not hold in our setting (we sketch a proof of this in Appendix \ref{app:1}). 

	The message passing model itself has been the subject of significant research interest over the past two decades. The majority of this work is concerned with \textit{exact} computation of Boolean functions of the inputs. Perhaps the canonical multi-party problem, and one which has strong applications to streaming, is set disjointness, where each player has a subset $S_i \subset [n]$ and the players want to know if $\cap_{i=1}^m S_i$ is empty. Bar-Yossef et al. \cite{bar2004information} demonstrated strong bounds for this problem in the black-board model. This lower bound resulted in improved (polynomially sized) lower bounds for streaming $F_p$ estimation for $p>2$. These results for disjointness have since been generalized and improved using new techniques \cite{chakrabarti2003near,gronemeier2009asymptotically, jayram2009hellinger, braverman2013tight}. 
Finally, we remark that while most results in the multi-party message passing model are not topology dependent, Chattopadhyay, Radhakrishnan, and Rudra have demonstrated that tighter topology-dependent lower bounds are indeed possible in the message passing model \cite{chattopadhyay2014topology}.



	\subsection{Road Map}
	
	In Section \ref{sec:prelims}, we formally introduce the  message passing model. In Section \ref{sec:messagepass}, we give our main $F_p$ estimation algorithm in the message passing model for $p >1$. In Section \ref{sec:HH}, we provide our algorithm for point-estimation and heavy hitters in the message passing model. In Section \ref{sec:pleq1} we give our $F_p$ estimation algorithm for $p<1$ in the message passing and streaming model. In Section \ref{sec:entropy}, we give our algorithm for entropy estimation, and in Section \ref{sec:amp} we give our algorithm for approximate matrix product.

	\section{Preliminaries}\label{sec:prelims}
	
	Let $f$ be a function $f:\R^n \to \R$. Let $G = (V,E)$ be a connected undirected graph with $m$ vertices, i.e. $V = \{1,\dots,m\}$. In the message passing model on the graph topology $G$, there are $m$ players, each placed at a unique vertex of $G$, with unbounded computational power. Player $i$ is given as input only a vector $X_i \in \Z^n$, which is known as the Number in Hand (NIH) model of communication. Let $\X = \sum_{i=1}^n X_i$ be the aggregate vector of the players inputs.  The goal of the players is to jointly compute or approximate the function $f(\X)$ by carrying out some previously unanimously agreed upon communication protocol. It is assumed that the graph topology of $G$ is known to all players.
	
		In this paper, we are concerned with the non-negative input model. Namely, the inputs $X_i$ satisfy $X_i \in \{0,1,\dots,M\}^n$ for all players $i$. Note an equivalent assumption to is that $(X_i)_j \geq 0$ for all $i$, and that the $(X_i)_j$'s can be specified in $O(\log(M))$ bits.
	 
	 \begin{remark}	
	For ease of presentation, we assume that $m,M = O(n^c)$ for some constant $c$. This allows us to simplify complexity bounds and write $\log(nmM) = O(\log n)$. This is a common assumption in the streaming literature, where $m$ corresponds to the length of the stream. We remark, however, that all our results hold for general $m,n,M$, by replacing each occurrence of $n$ in the communication complexity with $(mnM)$. 
	 \end{remark}
	
	 During execution of the protocol, a player $i \in V$ is only allowed to send a message to a player $j$ if $(i,j) \in E$. Thus, players may only communicate directly with their neighbors in the graph $G$. In contrast to the \textit{broadcast} and \textit{blackboard} models of communication, in the message passing model the message sent by player $i$ to player $j$ is only received by player $j$, and no other player. Upon termination of the protocol, at least one player must hold an approximation of the value $f(\X)$. For the protocols considered in this paper, this player will be fixed and specified by the protocol beforehand. We use $\mathcal{C} \in V$ to denote the distinguished player specified by the protocol to store the approximation at the end of the execution. 

	 Every such communication protocol in this model can be divided into rounds, where on the $j$-th round some subset $S_j \subseteq V$ of the players simultaneously send a message across one of their edges. Although it is not a restriction in the message passing model, our protocols satisfy the additional property that each player communicates \textit{exactly once}, across one of its edges, and that each player will receive messages from its neighbors in exactly one round. Specifically, for each player $i$, there will be exactly one round $j$ where some subset of its neighbors send player $i$ a message, and then player $i$ will send a single message in round $j+1$, and never again communicate. Such protocols are called \textit{one-shot} protocols.
	 
	 The\textit{ total communication} cost of a protocol is the total number of bits sent in all the messages during its execution. The \textit{max-communication} of a protocol is the maximum number of bits sent across any edge over the execution of the protocol. Communication protocols can be either deterministic or randomized. In this paper we consider the standard \textit{public-coin} model of communication, where each player is given shared access to an arbitrarily long string of random bits. This allows players to jointly utilize the same source of randomness without having to communicate it.

	
	Our protocols for $F_p$ estimation will utilize the $p$-stable distribution, $D_p$, which we will now introduce. For $p = 2$, the distribution $D_2$ is the just standard Gaussian distribution. Note for $p<2$, the distributions have heavy tails -- they decay like $x^{-p}$. Thus, for $p<2$, the variance is infinite, and for $p\leq 1$, the expectation is undefined. 
	\begin{definition}\label{def:stable}
		For $0 < p \leq 2$, there exists a probability distribution $D_p$ called the $p$-stable distribution. If $Z \sim D_p$, $p<2$, then the characteristic function of $D_p$ is given by $\ex{e^{itZ}	} = e^{-|t|^p}$. For $p=2$, $D_2$ is the standard Gaussian distribution. Moreover, for any $n$, and any $x \in \R^n$, if $Z_1,\dots,Z_n \sim D_p$ are independent, then $\sum_{i=1}^n Z_i x_i \sim \|x\|_p Z$, where $Z \sim D_p$, and $\sim$ means distributed identically to. 
	\end{definition}

	Standard methods for generating $p$-stable random variables are discussed in \cite{nolan2009stable}. Note that all protocols in this paper will generate these variables only to precision $1/\poly(n)$. 
	For a distribution $D_p$, we write $D_p^n$ to denote the product distribution of $D_p$. Thus $Z \sim D_p^n$ means $Z \in \R^n$ and $Z_1,\dots,Z_n$ are drawn i.i.d. from $D_p$. For reals $a,b\in \R$, we write $a = (1 \pm \eps)b$ to denote the containment $a \in [(1-\eps)b, (1+\eps)b]$. For an integer $t \geq 0$, we write $[t]$ to denote the set $\{1,2,\dots,t\}$.

	\section{Message Passing $F_p$ Estimation, $p > 1$} \label{sec:messagepass}

	In this section, we provide our algorithm for $F_p$ estimation, $1 \leq p \leq 2$, in the message passing model with max communication $O( \frac{1}{\eps^2}(\log d+ \log \log + \log 1/\eps))$, where $d$ is the diameter of the graph $G$.  
	 We begin by specifying the distinguished vertex $\mathcal{C} \in V$ which will hold and output the $F_p$ approximation at the end of the protocol. For a vertex $v \in G$, define its eccentricity $\ttx{ecc}(v) = \max_{u \in V}d(v,u)$, where $d(v,u)$ is the graph distance between $v,u$. We then set $\mathcal{C} \in V$ to be any vertex with minimal eccentricity. Such a vertex is known as a center of $G$. 	
	We now fix a shortest path spanning tree $T$ for $G$, rooted at the distinguished player $\mathcal{C}$. The spanning tree $T$ has the property that the path between $\mathcal{C}$ and any vertex $v \in V$ in the tree $T$ is also a shortest path between $\mathcal{C}$ and $v$ in $G$. Thus the distance between $\mathcal{C}$ and any vertex $v \in V$ is the same in $T$ as it is in $G$. The fact that the depth of $T$ is at most $d$, where $d$ is the diameter of $G$, now follows naturally. Such a shortest path spanning tree $T$ can be easily obtained via a breath first search.

Our algorithms for $F_p$ estimation and heavy hitters are based on a sketching step, followed by a randomized rounding procedure. Specifically, the players jointly generate a randomized matrix $S \in \R^{k \times n}$, such that one can first compute $\sum_i S X_i = S \X \in \R^k$, where $k << n$, and deduce the desired properties of $\X$ from this small sketch. Thus, for each coordinate $j \in [k]$, the players could all sketch their data $SX_i$, and send it up the tree $T$ to the distinguish vertex $\mathcal{C}$.
To improve the communication complexity required to send each coordinate of $SX_i$, the players randomly round the entries in $SX_i$ and send these instead, using only $O(\log \log(mM) + \log(1/\eps) + \log(d))$ bits of space. 

Before we introduce the randomized rounding protocol, we will need a technical Lemma about the behavior of $p$-stables. To prove it, we first use the following fact about the tails of $p$ stables, which can be found in \cite{nolan2009stable}.
\begin{proposition}\label{prop:stabletails}
 If $Z \sim D_p$ for $0 < p < 2$, then $\pr{|Z| \geq \lambda} \leq O(\frac{1}{\lambda^p})$.
\end{proposition}
\noindent
Also, we use the straightforward fact that $\|X_i\|_p^p \leq \|\sum_{i=1}^m X_i\|_p^p$ for non-negative vectors $X_i$ and $p \geq 1$.
\begin{fact}\label{fact1}
	If $X_1,\dots,X_m \in \R^n$ are entry-wise non-negative vectors and $1 \leq p \leq 2$, then $\sum_{i=1}^m \|X_i\|_p^p \leq \|\sum_{i=1}^m X_i\|_p^p$.
\end{fact}
\begin{proof}
	It suffices to consider the values coordinate-wise. If $x,y \geq 0$ and $1 \leq p \leq 2$, then $x^{p-1} \leq (x+y)^{p-1}$, so $x^p + y^p \leq x( x + y)^{p-1} + y(x+y)^{p-1} = (x+y)^{p}$, and the general result then follows from induction on the number of vectors $m$.
\end{proof}

This following lemma will be fundamental to our analysis. Recall by the $p$-stability of the distribution $D_p$, if $Z \sim D_p^n$ and $X_i \in \R^n$ is a vector, then we have $\langle Z,X_i\rangle \sim z \|X_i\|_p$, where $z \sim D_p$. Thus, we have then that $ \sum_{i=1}^m|\langle Z,X_i\rangle |^q  \sim \sum_{i=1}^m  z_i^q \|X_i\|_p^q$ for any $q > 0$ and vectors $X_1,\dots,X_m \in \R^n$. We would like to use the fact that for $1\leq p \leq q$, we have $\sum_{i=1}^m  \|X_i\|_p^q \leq \|\sum_{i=1}^m X_i\|_p^q = \|\X\|_p^q$ if the $X_i$'s are non-negative, and then bound the whole sum by $O(\|\X\|_p^q)$ with good probability. However, there are several complications. First, note that the $z_i$'s are not independent, since they are all generated by the same $Z$. Second, note that the $z_i$'s have heavy tails, so $\ex{z_i^q}$ will generally be infinite (e.g., if $p<2$ and $q=2$). Thus, we must be careful when attempting to bound the overall probability that this sum is large. To do so, we use a level-set analysis for variables with power-law tails.

\begin{lemma}\label{Lemma:lpVariance}
Fix $1 \leq p \leq q \leq 2$, and let $Z = (Z_1,Z_2,\dots,Z_n) \sim D_p^m$. Suppose $X_1,\dots,X_m \in \R^n$ are non-negative vectors, with $\X = \sum_j X_j$. Then for any $\lambda \geq 1$, if either $q-p \geq c > 0$ for some constant $c $ independent of $m$, or if $p=2$, we have
	\[ \bpr{ \sum_{j=1}^m |\langle Z, X_j \rangle |^q \geq C \lambda^{q} \|\X\|_p^q} \leq \frac{1}{\lambda^p} \]
	Otherwise, we have
		\[ \bpr{ \sum_{j=1}^m |\langle Z, X_j \rangle |^q \geq C \log(\lambda m) \lambda^{q} \|\X\|_p^q} \leq \frac{1}{\lambda^p} \]
	where $C$ is some constant (depending only on $c$ in the first case)\footnote{Observe that the extra $\log(m)$ factor is necessary in general when $p=q$. Note that if the support of all the $X_i's$ are distinct indicator vectors $e_i$, then the sum is the sum of the $p$-th powers of independent $p$-stables. If $p=1$, this is a sum of $m$ Cauchy random variables, which we expect to be $\Omega(m \log(m))$. For $1 \neq p < 2$ and $p=q,$ the tail behavior of $p$-stables raised to the $p$-th power is asymptotically similar to that of Cauchy's, so the result is the same here.}.	 
\end{lemma}
\begin{proof}

	First note, for $p=2$, each term $\sum_{j=1}^m |\langle Z, X_j \rangle |^2$ is distributed as $g^2 \|X_j\|_2^2$, where $g$ is Gaussian, thus $g^2$ is $\chi^2$ distributed with $\ex{g^2}  = 1$. It follows that $\ex{\sum_{j=1}^m |\langle Z, X_j \rangle |^2 } = \sum_{j=1}^m \|X_j\|_2^2 \leq \|\X\|_2^2$, where the inequality follows from the Fact \ref{fact1}, and the result then follows from Markov's inequality. 
	
Now suppose $p \in [1,2)$.	Again, by $p$-stability we have $ \sum_{j=1}^m |\langle Z, X_j \rangle |^q  = \sum_{j=1}^m \|X_j\|_p^q \hat{Z}_j^q$, where $\hat{Z}_j \sim D_p$. Note though that the $\hat{Z}_j$'s are not independent. Define $I_k = \{ j \in [m] \;|\;2^{-k} \|\X\|_p^q  \leq \hat{Z}_j^q \|X_j\|_p^q	\leq 2^{-k+1} \|\X\|_p^q  \}$.  So if $j \in I_k$, then we have $\hat{Z}_j^q \geq 2^{-k} \frac{\|X\|_p^q }{|X_j|_p^q}$. For each $k$, the goal is to bound the contribution of the terms in the set $I_k$ to the overall sum. By Proposition \ref{prop:stabletails}, there is some constant $c \geq 1$ such that
\begin{equation}
\begin{split}
	\pr{j \in I_k} &\leq c\left(2^{k} \frac{\|X_j\|_p^q}{ \|\X\|_p^q}\right)^{p/q} \\
	&= c2^{pk/q} \frac{\|X_j\|_p^p}{ \|\X\|_p^p}\\
\end{split}
\end{equation}
	 By Fact \ref{fact1}, we have $\sum_j \|X_j\|_p^p \leq \|\X\|_p^p$. 
So we obtain $\ex{|I_k|} \leq  c2^{pk/q}\sum_j \frac{\|X_j\|_p^p}{\|\X\|_p^p} \leq c2^{pk/q}$ for any $k \in \Z$.

Now consider the event $\mathcal{E}_0$ that $I_k = \emptyset$ for all $k \leq - \frac{q}{p}\log( 10C'\lambda)$. We first analyze $\pr{\mathcal{E}_0}$. Note that for a fixed $j$, the probability that $j \in I_k$ for some $k \leq -\frac{q}{p}\log(10C' \lambda)$ is at most $ \lambda^{-1} \frac{\|X_j\|_p^p}{ \|\X\|_p^p}/10$. So the expected number of $j \in [n]$ with $j \in I_k$ for some $k \leq - \frac{q}{p}\log(10 C'\lambda)$ is at most $\lambda^{-1} \sum_j \frac{\|X_j\|_p^p}{ \|\X\|_p^p}/10 \leq  \lambda^{-1}/10$. So by Markov's inequality, we conclude that $\pr{\mathcal{E}_0} \leq \frac{1}{10\lambda}$.
Now note that for any $k \in \Z$, the contribution of the items $j \in I_k$ is at most $|I_k|2^{-k + 1}\|\X\|_p^q$. Thus, an upper bound on the expected contribution of the items in sets $I_k$ for $ - \frac{q}{p}\log( 10C'\lambda) \leq k \leq   4 \log(m \lambda )$ is given by

	\[\sum_{j= - \frac{q}{p}\log( 10C'\lambda) }^{ 4 \log(m \lambda )}\bex{|I_k| 2^{-k + 1}\|\X\|_p^q } \leq \sum_{j= - \frac{q}{p}\log( 10C'\lambda) }^{ 4 \log(m \lambda )}c 2^{-k(1 - p/q)} \|\X\|_p^q \]
Now if $q-p >c'$ for some constant $c'$ independent of $m$, then the above sum is geometric, and at most $O(2^{ (\frac{q}{p} - 1)\log(10C'\lambda)  } \|\X\|_p^q ) = O( (10 C' \lambda)^{q/p - 1} \|\X\|_p^q)$. If $p/q$ is arbitrarily small, the above can be bounded by $O(\log(m \lambda) (10 C' \lambda)^{q/p - 1} \|\X\|_p^q)$. Setting $\delta = O(\log (m\lambda ))$ if this is the case, and $\delta = 1$ otherwise, we can apply Markov's to obtain:

\[ \bpr{\sum_{j= - \frac{q}{p}\log( 10C'\lambda) }^{ 4 \log(m \lambda )}	|I_k| 2^{-k + 1}\|\X\|_p^q >  \delta(10 C' \lambda)^{q/p} \|\X\|_p^q } < \frac{1}{10 C' \lambda}\]

Call the above event $\mathcal{E}_1$. Conditioned on $\mathcal{E}_0 \cup \mathcal{E}_1$ not occurring, which occurs with probability at least $1 - \frac{1}{\lambda}$, it follows that the contribution of the items in level sets $I_k$ for $k \leq 4 \log(m \lambda)$ is at most $O(\delta \lambda^{q/p} \|\X\|_p^q)$. Now for $k \geq 4 \log(m \lambda)$, note that the contribution of any term $j \in I_k$ is at most $\frac{1}{m^2} \|\X\|_p^2$, and so the contribution of all such items to the total sum is at most $\frac{1}{m}\|\X\|_p^2$. Thus $ \sum_{j=1}^m |\langle Z, X_j \rangle |^q  = O(\delta \lambda^{q/p} \|\X\|_p^q)$ with probability at least $1/\lambda$, which is the desired result after replacing $\lambda$ with $\lambda^{p}$ as needed.

\end{proof}

While the above result was described to specifically hold only for $p$-stable random variables, it is straightforward to show that the result for $p=2$ holds when Gaussian random variables are replaced with Rademacher random variables: i.e., variables $Z_i$ that are uniform over $\{1,-1\}$.
\begin{corollary}\label{cor:rademacher}
	Suppose $Z = (Z_1,\dots,Z_m)$ where the $Z_i$'s are uniform over $\{1,-1\}$ and pairwise independent, and let $X_1,\dots,X_m$ be non-negative vectors with $\X = \sum_j X_j$. Then for any $\lambda \geq 1$, we have
		\[ \bpr{ \sum_{j=1}^m |\langle Z, X_j \rangle |^2 \geq  \lambda \|\X\|_2^2} \leq \frac{1}{\lambda} \]
\end{corollary}
\begin{proof}
	We have $\ex{|\langle Z, X_j \rangle |^2} = \|X_j\|_2^2$, so the result follows from an application of Markov's inequality and using the fact that $\sum_j\|X_j\|_2^2 \leq \|\X\|_2^2$.
\end{proof}

We now note that for $p=2$, we obtain much stronger concentration. We will need this for the approximate matrix product section.
\begin{corollary}\label{cor:p2highprob}
Let $Z = (Z_1,Z_2,\dots,Z_n) \sim D_2^m$ be i.i.d. Gaussian. Suppose $X_1,\dots,X_m \in \R^n$ are non-negative vectors, with $\X = \sum_j X_j$. Then for any $\lambda \geq c\log(m)$ for some sufficiently large constant $c$, we have
	\[ \bpr{ \sum_{j=1}^m |\langle Z, X_j \rangle | \geq  \lambda \|\X\|_2^2} \leq \exp(-C\lambda) \]
	where $C$ is some universal constant.	
\end{corollary}
\begin{proof}
	Recall for $j \in [m]$, $|\langle Z, X_j \rangle |^2$ is distributed as $g_j^2 \|X_j\|_2^2$, where $g_j$ is Gaussian, thus $g_j^2$ is $\chi^2$ distributed with $\ex{g_j^2}  = 1$. Now if $z$ is $\chi^2$ distributed, we have $\pr{z > \lambda} < e^{-C\lambda}$ for some absolute constant $C$. So   $\pr{|\langle Z, X_j \rangle |^2 >\log(m)\lambda } < e^{-C\lambda}/m$ for a slightly different constant $C$. We can union bound over all $m$ terms, so we have $\sum_{j=1}^m |\langle Z, X_j \rangle |^2 \leq \sum_{j=1}^m \lambda \|X_j\|_2^2 \leq \lambda \|\X\|_2^2$ with probability $1- \exp(-C\lambda)$ as needed.
\end{proof}

\subsection{Randomized Rounding of Sketches}\label{sec:randomround}

We now introduce our randomized rounding protocol.	Consider non-negative integral vectors $X_1,X_2,\dots,X_m \in \Z^n_{\geq 0}$, with $\X = \sum_{i=1}^n X_i$. Fix a message passing topology $G = (V,E)$, where each player $i \in V$ is given as input $X_i$. Fix any vertex $\mathcal{C}$ that is a center of $G$, and let $T$ be a shortest path spanning tree of $G$ rooted at $\mathcal{C}$ as described at the beginning of the section. Let $d$ be the depth of $T$. The players use shared randomness to choose a random vector $Z \in \R^n$, and their goal is to approximately compute $\langle Z, \X\rangle = \langle Z, \sum_{i=1}^m X_i \rangle$. 
	The goal of this section is to develop a $d$-round randomized rounding protocol, so that at the end of the protocol the approximation to  $\langle Z, \X\rangle$ is stored at the vertex $\mathcal{C}$.

	We begin by introducing the rounding primitive which we use in the protocol. Fix $\eps >0$, and let  $\gamma = (\eps \delta /\log(nm))^C$, for a sufficiently large constant $C>1$. For any real value $r \in \R$, let $i_r \in \Z$ and $\alpha_i \in \{1,-1\}$ be such that $(1+ \gamma)^{i_r}  \leq \alpha_i  r \leq (1+ \gamma)^{i_r+1}$.  Now fix $p_r$ such that:
	\[\alpha_i r = p_r (1+ \gamma)^{i_r+1} + (1-p_r) (1+ \gamma)^{i_r}  \]
	We then define the rounding random variable $\Gamma(r)$ by
	\[ \Gamma(r) =  \begin{cases}	0 & \text{if } r=0 \\ 
	 \alpha_i(1+ \gamma)^{i_r+1} & \text{with probability } p_r\\
	\alpha_i(1+ \gamma)^{i_r} & \text{with probability } 1-p_r \\	
	\end{cases}\]

	The following proposition is clear from the construction of $p_r$ and the fact that the error is deterministically bounded by $\gamma|r|$.
\begin{proposition}
For any $r \in R$,	We have  $\ex{\Gamma(r)} = r$ and $\var{\Gamma(r)} \leq r^2 \gamma^2$
\end{proposition}

 We will now describe our rounding protocol.
 We partition $T$ into $d$ layers, so that all nodes at distance $d-t$ from $\mathcal{C}$ in $T$ are put in layer $t$. Define $L_t \subset [n]$ to be the set of players at layer $t$ in the tree. For any vertex $u \in G$, let $T_u$ be the subtree of $T$ rooted at $u$ (including the vertex $u$). For any player $i$, let $C_i \subset [n]$ be the set of children of $i$ in the tree $T$. The procedure for all players $j \in V$ is then given in Figure \ref{alg:rounding}.

\begin{center}
	
	\begin{figure}[h!]
		\fbox{\parbox{\textwidth}{ Procedure for node $j$ in layer $i$:
				
				\begin{enumerate}[topsep=0pt,itemsep=-1ex,partopsep=1ex,parsep=1ex] 
					\item Choose random vector $Z \in \R^{n}$ using shared randomness. 
					\item Receive rounded sketches $r_{j_1},r_{j_2},\dots,r_{j_{t_j}} \in \R$ from the $t_j$ children of node $j$ in the prior layer (if any such children exist). 
					\item Compute $x_j = \langle X_j, Z\rangle+ r_{j_1} + r_{j_2} + \dots + r_{j_t} \in \R$.
					\item Compute $r_j =\Gamma(x_j)$. If player $j \neq \mathcal{C}$, then send $r_j$ it to the parent node of $j$ in $T$. If $j = \mathcal{C}$, then output $r_j$ as the approximation to $\langle Z,\X\rangle$.
				\end{enumerate}
		}}\caption{Recursive Randomized Rounding} \label{alg:rounding}
	\end{figure}
	
\end{center}

 For each player $i$ in layer $0$, they take their input $X_i$, and compute $\langle Z,X_i\rangle$. They then round their values as $r_i = \Gamma(\langle Z,X_i\rangle)$,  where the randomness used for the rounding function $\Gamma$ is drawn independently for each call to $\Gamma$. Then player $i$ sends $r_i$ to their parent in $T$. In general, consider any player $i$ at depth $j > 0$ of $T$. At the end of the $j$-th round, player $i$ will receive a rounded value $r_\ell$ for every child vertex $\ell \in C_i$. They then compute $x_i = \langle Z,X_i\rangle + \sum_{\ell \in C_i} r_\ell$, and $r_i = \Gamma(x_i)$, and send $r_i$ to their parent in $T$. This continues until, on round $d$, the center vertex $\mathcal{C}$ receives $r_\ell$ for all children $\ell \in C_\mathcal{C}$. The center $\mathcal{C}$ then outputs $r_\mathcal{C} =\langle Z, X_\mathcal{C} \rangle + \sum_{\ell \in C_\mathcal{C}} r_\ell$ as the approximation.

We will now demonstrate that if the entries $Z_i$ of $Z$ are drawn independently from a $p$-stable distribution, we obtain a good estimate of the product $\langle Z, \X \rangle$.  For any player $i$, let $Q_i =  \sum_{u \in T_i} X_u$, and
$y_i =  \langle Z, Q_i \rangle$. Then define the error $e_i$ at player $i$ as $e_i = y_i - r_i$. We first prove a proposition that states the expectation of the error $e_i$ for any player $i$ is zero, and then the main lemma which bounds the variance of $e_i$. The error bound of the protocol at $\mathcal{C}$ then results from an application of Chebyshev's inequality.

\begin{proposition}
	For any player $i$, we have $\ex{e_i} = 0$. Moreover, for any players $i,j$ such that $i \notin T_j$ and $j \notin T_i$, the variables $e_i$ and $e_j$ are statistically independent. 
\end{proposition}
\begin{proof}
We prove the first claim by induction on the layer $j$ such that $i \in L_j$. For the base case of $j=0$, player $i$ has no children, so $e_i =  r - \Gamma(r)$ for some $r \in \R$, from which the result follows from the fact that $\Gamma(r)$ being an unbiased estimate of $r$. Now supposing the result holds for all players in $L_j$, and let $i \in L_{j+1}$ with children $i_1,\dots,i_k \in [n]$. Then 

\begin{equation}
\begin{split}
\bex{e_i} & = \bex{y_i - \Gamma(y_i + e_{i_1} + \dots + e_{i_k})} \\
 &=  \mathbb{E}_{e_{i_1}, \dots, e_{i_k}} \left[ \bex{y_i - \Gamma(y_i + e_{i_1} + \dots + e_{i_k}) \;\ \big| \; e_{i_1},\dots,e_{i_k}} 	\right] \\ 
  &=  \mathbb{E}_{e_{i_1}, \dots, e_{i_k}} \left[e_{i_1} + \dots + e_{i_k}	\right] \\ 
  &=0\\
\end{split}                                           
\end{equation}
	which completes the first claim. The second claim follows from the fact that the randomness used to generate $e_i$ and $e_j$ is distinct if the subtrees of player $i$ and $j$ are disjoint. 
\end{proof}

\begin{lemma}\label{lemma:roundingmain}
	Fix $p \in [1,2],$ and let $Z = (Z_1,Z_2,\dots,Z_n)  \sim D_p^n$. Then the above procedure when run on $\gamma = (\eps \delta /(d\log(nm)))^C$ for a sufficiently large constant $C$, produces an estimate $r_\mathcal{C}$ of $\langle Z, \X \rangle$, held at the center vertex $\mathcal{C}$, such that $\ex{r_\mathcal{C}} = \langle Z,\X\rangle$. Moreover, over the randomness used to draw $Z$, with probability $1-\delta$ for $p<2$, and with probability $1-e^{-1/\delta}$ for Gaussian $Z$, we have $\ex{( r_\mathcal{C} -\langle Z,\X\rangle )^2} \leq (\eps/\delta)^2 \|\X\|_p$. Thus, with probability at least $1-O(\delta)$, we have
	 \[ |r_\mathcal{C} - \langle Z, \X \rangle| \leq \eps \|\X\|_p\]
	 Moreover, if $Z = (Z_1,Z_2,\dots,Z_n) \in \R^n$ where each $Z_i \in \{1,-1\}$ is a 4-wise independent Rademacher variable, then the above bound holds with $p=2$ (and with probability $1-\delta$).
\end{lemma} 
\begin{proof}

	Set $N = \poly(n,m)^{\poly(d/(\delta\eps))}$ to be sufficiently large, and let $\gamma =(\eps\delta  /(d \log(nm)))^C$ so that $1/\gamma$ is a  sufficiently large polynomial in $\log(N)$.  Let $\gamma_0 = \gamma (\frac{ \log(N)}{\eps})^c$ for a constant $c < C$ that we will later choose. 
For the remainder of the proof, we condition on the event $\mathcal{E}_i$ for $i \in [d]$ that $\sum_{i \in L_t} y_i^2 \leq \log^{2}(N) \|\sum_{i \in L_t} Q_i \|_p^2$. By Lemma \ref{Lemma:lpVariance} (or Corollary \ref{cor:rademacher} for Rademacher variables) and a union bound, we have $\pr{\cup_{j \in [d]} \mathcal{E}_j} \leq  \log^{-1}(N)$. For Gaussian $Z$, we have $\pr{\cup_{j \in [d]} \mathcal{E}_j} \leq  1/N \leq \exp(-1/\delta)$ by Corollay \ref{cor:p2highprob}.  Note that $\mathcal{E}_i$ depends only on the randomness used to sample $Z$, and not on the randomnesses used in the rounding protocol.

	We now prove by induction that for any player $i \in L_j$, we have $\ex{e_i^2} \leq (j+1) \gamma_0^2 \sum_{v \in T_i} | \langle Q_v , Z \rangle|^2$. For the base case, if $i \in L_0$ by definition of the rounding procedure we have $e_i^2 \leq \gamma^2  | \langle Q_i , Z \rangle|^2 \leq \gamma_0^2  | \langle Q_i , Z \rangle|^2$. Now suppose the result holds for layer $j$, and let $i \in L_{j+1}$, and let $C_i \subset [n]$ be the children of player $i$. Then $e_i = \eta_i + \sum_{v \in C_i} e_v$, where $\eta_i$ is the error induced by the rounding carried out at the player $i$. Then $\eta_i$ is obtained from rounding the value $\left(\langle Q_i , Z \rangle +  \sum_{v \in C_i} e_v\right)$, thus 
	\begin{equation}
	\begin{split}
	\ex{\eta_i^2} &\leq \gamma^2	\bex{ \left(\langle Q_i , Z \rangle + \sum_{v \in C_i} e_v\right)^2}\\
	& = \gamma^2 \left(\left|\langle Q_i , Z \rangle\right|	+ \bex{\sum_{v \in C_i} e_v^2}\right) \\
		& \leq \gamma^2 \left( \left|\langle Q_i , Z \rangle\right|	+ (j+1)\gamma_0^2\sum_{v \in C_i}  \sum_{u \in T_v} | \langle Q_u, Z \rangle|^2 \right)\\
	\end{split}
	\end{equation}
Where the first equality holds by the independence and mean $0$ properties of the errors $e_v$, and the last inequality holds by induction. Thus
		\begin{equation}
	\begin{split}
	\ex{e_i^2} & = \ex{(\eta_i + \sum_{v \in C_i} e_v)^2 )}\\
	& = \ex{\eta_i^2} + \ex{\sum_{v \in C_i} e_v^2 } \\
	& \leq \gamma^2 \left( \left|\langle Q_i , Z \rangle\right|	+ (j+1)\gamma_0^2\sum_{v \in C_i}  \sum_{u \in T_v} | \langle Q_u, Z \rangle|^2 \right)  + \ex{\sum_{v \in C_i} e_v^2 }\\
	&\leq \gamma_0^2  \left|\langle Q_i , Z \rangle\right|	+ \gamma_0^2\sum_{v \in C_i}  \sum_{u \in T_v} | \langle Q_u, Z \rangle|^2  + (j+1) \gamma_0^2 \sum_{v \in C_i}\sum_{u \in T_v} | \langle Q_u, Z \rangle|^2   \\
	&\leq (j+2) \gamma_0^2  \left|\langle Q_i , Z \rangle\right|	  + (j+2) \gamma_0^2 \sum_{v \in C_i}\sum_{u \in T_v} | \langle Q_u, Z \rangle|^2   \\
	&=  (j+2) \gamma_0^2 \sum_{u \in T_i}  | \langle Q_u, Z \rangle|^2 \\
	\end{split}
	\end{equation}
which is the desired result. The variance bound then follows after conditioning on $\cap_j \mathcal{E}_j$. It follows that $|r_\mathcal{C} - \langle Z, \X \rangle|^2 \leq \frac{2}{\delta} \gamma_0^2 (d+1)\sum_{u \in T}  | \langle Q_u, Z \rangle|^2$ with probability at least $1-\delta/2$ by Chebyshev's inequality. We now condition on this and $\cap_j \mathcal{E}_j$, which hold together with probability $1-\delta$ by a union bound. Since we conditioned on $\cap_j \mathcal{E}_j$, we have 
	\begin{equation}
\begin{split}
|r_\mathcal{C} - \langle Z, \X \rangle|^2 &\leq \frac{1}{\delta} \gamma_0^2 (d+1)\sum_{u \in T}  | \langle Q_u, Z \rangle|^2 \\
& =  \frac{1}{\delta} \gamma_0^2 (d+1)\sum_{i=1}^d \sum_{u \in L_i} y_u \\
& \leq \frac{1}{\delta} \gamma_0^2 (d+1)\sum_{i=1}^d \log^2(N) \|\sum_{u \in L_i} Q_u \|_p^2 \\
& \leq \frac{1}{\delta} \gamma_0^2 (d+1)\sum_{i=1}^d \log^2(N) \| \X\|_p^2 \\
& \leq \eps^2 \| \X\|_p^2 \\
\end{split}
\end{equation}
	which completes the proof. Note here, in the second to last inequality, we used the fact that since $Q_u$ are positive vectors, for $p \leq 2$ we have $\sum_{i=1}^d  \|\sum_{u \in L_i} Q_u \|_p^2  \leq  \|\sum_{i=1}^d  \sum_{u \in L_i} Q_u \|_p^2 \leq \|\X\|_p^2$.

\end{proof}

\begin{theorem}\label{thm:LPmain}
	For $p \in (1,2]$, there is a protocol for $F_p$ estimation which succeeds with probability $3/4$ in the message passing model, which uses a total of $O(\frac{m}{\eps^2}(\log(\log(n)) + \log(d) + \log(1/\eps) ))$ communication, and a max communication of $O(\frac{1}{\eps^2}(\log(\log(n)) + \log(d) + \log(1/\eps) ))$, where $d$ is the diameter of the communication network. 
\end{theorem}
\begin{proof}
	We use Indyk's classic $p$-stable sketch $S \in \R^{k \times n}$, where $k = \Theta(1/\eps^2)$, and each entry $S_{i,j} \sim D_p$. It is standard that computing $\ttx{median}_i \{ |(S \X)_i| \}$ gives a $(1 \pm \eps)$ approximation to $\|\X\|_p$ with probability $3/4$ \cite{indyk2006stable}. By running the randomized rounding protocol on each inner product $\langle S_j, \X\rangle$ independently with error parameter $\eps' = \poly(\eps)$ and failure parameter $\delta' = \poly(\eps )$ (rounding to powers of $(1 + \gamma)$ where $\gamma$ is as defined in this section in terms of $\eps',\delta',d$), it follows that the central vertex $\mathcal{C}$ recovers $S\X$ up to entry-wise additive $\eps' \|\X\|_p$ error with probability $1-O(\eps)$, after union bounding over all $k$ repetitions using Lemma \ref{lemma:roundingmain}. This changes the median by at most $\eps' \|X\|_p$, which results in an  $\eps'\|\X\|_p$ additive approximation of $\|\X\|_p$, which can be made into a $\eps \|\X\|_p$ additive approximation after a rescaling of $\eps$. 
	
	For the communication bound, note that for each coordinate of $S\X$, exactly one message is sent by each player. Set $K =  Mnm/\gamma$, where $\gamma$ is as in the rounding procedure.  By Proposition \ref{prop:stabletails}, we can condition on the fact that $|Z_i| \leq c K^3$ for all $i \in [n]$ and for some constant $c>0$, which occurs with probability at least $1-1/K^2$. Now by the proof of Lemma \ref{lemma:roundingmain}, we have that $\ex{e_i^2} \leq (j+1) \gamma_0^2 \sum_{v \in T_i} | \langle Q_v , Z \rangle|^2 \leq K^5$, where $e_i = \sum_{u \in T_i} X_u - r_i$, where $r_i$ is the message sent by the $i$-th player for a single coordinate of $SX_i$. By Markov's with probability $1-1/K^2$ we have $|e_i| < K^4$, and thus $|r_i| \leq K^6$ for all $i$.

Now for any $r_i$ for player $i$ in layer $\ell$ with $|r_i| < 1/(mK)^{d+3-\ell}$, we can simply send $0$ instead of $r_i$. Taken over all the potential children of a player $j$ in layer $\ell+1$, this introduces a total additive error of $1/K^{d+3-\ell}$ in $x_j$. Now if $x_j$ originally satisfies $|x_j| > 1/K^{d+2-\ell}$, then the probability that this  additive error of $1/K^{d+3-\ell}$ changes the rounding result $r_j = \Gamma(x_j)$ is $O(1/K)$, and we can union bound over all $m$ vertices that this never occurs. Thus, the resulting $r_j$ is unchanged even though $x_j$ incurs additive error. Otherwise, if  $|x_j| \leq 1/K^{d+2-\ell}$, then since Player $j$ is in layer $(\ell+1)$, we round their sketch $x_j$ down to $0$ anyway. The final result is an additive error of at most $1/K^2$ to $r_{\mathcal{C}}$. Note that we can tolerate this error, as it is only greater than $\gamma\|\X\|_p$ when $\X = 0$, which is a case that can be detected with $O(1)$ bits of communication per player (just forwarding whether their input is equal to $0$). With these changes, it follows that $ 1/(mK)^{d+3} \leq r_j \leq K^6$ for all players $j$. Thus each message $r_j$ can be sent in $O( \log ( \log((mK)^{d+3} ) )) = O(\log \log(n) + \log(d) + \log(1/\eps))$ as needed.

\end{proof}

\subsection{Heavy Hitters and Point Estimation}\label{sec:HH}
In this section, we show how our randomized rounding protocol can be used to solve the $L_2$ heavy hitters problem. For a vector $\X \in \R^n$, let $\X_{\text{tail}(k)}$ be $\X$ with the $k$ largest (in absolute value) entries set equal to $0$. Formally, given a vector $\X \in \R^n$, the heavy hitters problem is to output a set of coordinates $H \subset[n]$ of size at most $|H| = O(\eps^{-2})$ that contains all  $i \in [n]$ with $|\X_i| \geq \eps \|\X_{tail(1/\eps^2)}\|_2$. Our protocols solve the strictly harder problem of \textit{point-estimation}. The point estimation problem is to output a $\tilde{\X} \in \R^n$ such that $\|\tilde{\X} - \X\|_\infty \leq \eps  \|\X_{tail(1/\eps^2)}\|_2$.
Our protocol uses the well-known \textit{count-sketch} matrix $S$ \cite{charikar2002finding}, which we now introduce. 

\begin{definition}	\label{def:countsketch}
Given a precision parameter $\eps$ and an input vector $\X \in \R^n$, count-sketch stores a table $A \in \R^{\ell \times 6/\eps^2}$, where $\ell = \Theta(\log(n))$. Count-sketch first selects pairwise independent hash functions $h_j:[n]\to [6/\eps^2]$ and $4$-wise independent $g_j:[n] \to \{1,-1\}$, for $j=1,2,\dots,\ell$. Then for all $i \in [\ell], \:j \in [6/\eps^2]$, it computes the following linear function $A_{i,j} = \sum_{k \in [n], h_i(k) = j} g_i(k) \X_k$, and outputs an approximation $\tilde{\X}$ of $\X$ given by $$\tilde{\X_k} = \text{median}_{i \in [\ell]} \{g_i(k) A_{i,h_i(k)}\}$$\end{definition}

 Observe that the table $A \in \R^{\ell \times 6/\eps^2}$ can be flattened into a vector $A \in \R^{6\ell /\eps^2}$. Given this, $A$ can be represented as $A = S\X$ for a matrix $S \in \R^{6\ell/\eps^2 \times n}$. The matrix for any $i \in [\ell]$ and $j \in [6/\eps^2]$ and $\ell \in [n]$, the matrix $S$ is given by  $S_{(i-1)(6 /\eps^2) + j, \ell} = \delta_{i,j,\ell} g_{j}(\ell)$, where $\delta_{i,j,\ell}$ indicates the event that $h_i(\ell) = j$. Alternatively, the matrix $S$ can be built as follows. Define $S_i \in \R^{6/\eps^2 \times n}$ by 
 \[ (S_i)_{p,q} = \begin{cases} g_i(q) & \text{if } h_i(q) = p \\ 0 & \text{otherwise} \end{cases}\]
 Then $S_i$ has exactly one non-zero entry per column, and that entry is in a random row (chosen by the hash function $h_i$), and is randomly either $1$ or $-1$ (chosen by $g_i$). The whole matrix  $S \in \R^{6\ell/\eps^2 \times n}$ can be built by stacking the matrices $S_1,\dots,S_\ell \in \R^{6/\eps^2 \times n}$ on top of each other. Given $S\X$ with $k = \Theta(\frac{1}{\eps^2}\log(n))$, one can solve the point-estimation problem as described in Definition \ref{def:countsketch} \cite{charikar2002finding}. In order to reduce the communication from sending each coordinate of $S\X$ exactly, we use our rounding procedure to approximately compute the sketch $S\X$. 

We remark that there is a $O(\frac{1}{\eps^2} \log(n))$-space streaming algorithm for the heavy hitters problem in insertion only streams, known as BPTree \cite{braverman2016bptree}. However, this streaming algorithm does not produce a linear sketch, and is not mergeable. The BPTtree algorithm crucially relies on sequentially learning one bit of the identity of each heavy hitter at a time. However, in the message passing model, the data cannot be sequentially accessed unless a path $P$ was fixed running through all the players in $G$. Such a path may cross the same edge more than once, thus cannot be one-way, and also will require as many as $m$ rounds (instead of the $d$ given by our algorithm). Thus constructions such as BPTtree cannot be used as is to solve the heavy hitters problem in the message passing model. Moreover, BPTtree does not solve the frequency estimation problem, which our protocol does in fact accomplish.

\begin{theorem}\label{thm:HHmain}
	Consider a message passing topology $G = (V,E)$ with diameter $d$, where the $i$-th player is given as input $X_i \in \Z^n_{\geq 0}$ and $\X = \sum_{i=1}^m X_i$.  Then there is a communication protocol which outputs an estimate $\tilde{\X} \in \R^n$ of $\X$ such that 
	\[ \|\tilde{\X} - \X\|_\infty \leq \eps \|\X_{\text{tail}(1/\eps^2)}\|_2 \]
	with probability $1-1/n^c$ for any constant $c \geq 1$. The protocol uses  $O(\frac{m}{\eps^2}\log(n)(\log(\log(n)) + \log(d) + \log(1/\eps) ))$ total communication, and a max communication of $O(\frac{1}{\eps^2}\log(n)(\log(\log(n)) + \log(d) + \log(1/\eps) ))$.
\end{theorem}
\begin{proof}
	Note that the non-zero entries of each row of $S$ are $4$-wise independent Rademacher variables. So by Lemma \ref{lemma:roundingmain}, for each entry $(S\X)_j$ of $S\X$, we can obtain an estimate $r_{\mathcal{C}}$ at the central vertex $\mathcal{C}$ such that $|r_{\mathcal{C}} - (S\X)_j| \leq \eps \|\X_{\text{supp}(j)}\|_2$, where $\text{supp}(j) \subset [n]$ is the support of the $j$-th row of $S$.
	Let $A \in \R^{d \times 6/\eps^2}$ be the table stored by count-sketch as in Definition \ref{def:countsketch}. We first go through the standard analysis of count-sketch. Consider the estimate of a fixed entry $k \in [n]$. Consider a given entry $g_i(k) A_{i,h_i(k)}$ for a given row $i$ of $A$. With probability $5/6$, none of the top $1/\eps^2$ largest items in $\X$ (in absolute value) collide with $k \in [n]$ in this entry. Call this event $\mathcal{E}$, and condition on it now. Now note that $\ex{g_i(k) A_{i,h_i(k)}} = \X_k$, and $\ex{(g_i(k) A_{i,h_i(k)})^2} = \ex{\sum_{\ell: h_i(\ell) = h_i(k)} \X_\ell^2} + \ex{\sum_{u \neq v: h_i(u) = h_i(v) = h_i(k)} \X_u \X_v g_i(u) g_i(v)  } = \ex{\sum_{\ell: h_i(\ell) = h_i(k)} \X_\ell^2}$, which is at most $\eps^2\|\X_{\text{tail} (1/\eps^2)}\|_2^2/6$ after conditioning on $\mathcal{E}$. Thus with probability $5/6$, we have $|g_i(k) A_{i,h_i(k)} - \X_k | \leq \eps  \|\X_{\text{tail}(1/\eps^2)}\|_2$. Note that we conditioned on $\mathcal{E}$, so altogether we have that $|g_i(k) A_{i,h_i(k)} - \X_k | \leq \eps  \|\X_{\text{tail}(1/\eps^2)}\|_2$ with probability at least $2/3$. Thus with probability at least $1-1/n^2$, the above bound holds for the median of the estimates $\text{median}_{i \in [d]} \{g_i(k) A_{i,h_i(k)}\}$ , and we can then union bound over all $k \in [n]$ to obtain the desired result.
	
	Now consider the changes that occur to this argument when we add an additional $\eps \|\X_{\text{supp}(j)}\|_2$ error to each $A_{i,h_i(k)}$. Let $r_{\mathcal{C}}^{i,h_i(k)}$ be the estimate held by the central vertex of $A_{i,h_i(k)}$. As noted above, after conditioning on $\mathcal{E}$, we have $\ex{\|\X_{\text{supp}(j)}\|_2^2} \leq  \eps^2 \|\X_{\text{tail}(1/\eps^2)}\|_2^2$, so with probability $15/16$ we have  $\|\X_{\text{supp}(j)}\|_2 \leq 4\eps  \|\X_{\text{tail}(1/\eps^2)}\|_2$. Taken together with the event that $|g_i(k) A_{i,h_i(k)} - \X_k | \leq \eps  \|\X_{\text{tail}(1/\eps^2)}\|_2$, which occurred with probability $5/6$ after conditioning on $\mathcal{E}$,  it follows that $|g_i(k)r_{\mathcal{C}}^{i,h_i(k)} - \X_k| \leq  \eps  \|\X_{\text{tail}(1/\eps^2)}\|_2$ with probability $1 - (1/6 + 1/6 + 1/16) > 3/5$. The result now follows now by the same argument as above (note that it only required this probability to be at least $1/2 + \Omega(1)$).
	
	The message complexity analysis is identical to that of Theorem \ref{thm:LPmain}, except instead of applying Proposition \ref{prop:stabletails} to bound the probability that $|Z_i| \leq c (Mnm/\gamma)^3$ for all $i \in [n]$ (Where $Z$ is the non-zero entries in some row of $S$), we can use the deterministic bound that $|Z_i| \leq 1$. By the same argument as in Theorem \ref{thm:LPmain}, each rounded message $r_j$ requires $O(\log \log(n) + \log(d) + \log(1/\eps))$ bits to send with high probability. Since each player sends $6d/\eps^2 = O(\log(n)/\eps^2)$ messages (one for each row of $S$), the claimed communication follows.

\end{proof}

	\section{$F_p$ Estimation for $p<1$} \label{sec:pleq1}
		In this section, we develop algorithms for $F_p$ estimation for $p<1$ in the message passing model, and in the process obtain improved algorithms for entropy estimation. Our algorithms require a max communication of $O(\frac{1}{\eps^2}( \log\log n + \log 1/\eps))$, which is independent of the diameter of the graph topology. In particular, these results hold for the directed line graph, where communication is only allowed in one direction. As a consequence, we improve upon the best known algorithms for the space complexity of insertion only streaming algorithms.
	
	We begin by reviewing the fundamental sketching procedure used in our estimation protocol. The sketch is known as a Morris counter. We point out that, in fact, our rounding algorithm from Section \ref{sec:randomround} reduces to being a Morris counter when run on insertion only streams. 
	
	\subsection{Morris Counters}
	We begin by describing the well known approximate counting algorithm, known as a Morris Counter \cite{morris1978counting,flajolet1985approximate}. 
	The algorithm first picks a base $1 < b \leq 2$, and initalizes a counter $C \leftarrow 0$. Then, every time it sees an insertion, it increments the counter $C \leftarrow C+\delta$, where $\delta = 1$ with probability $b^{-C}$, and $\delta = 0$ otherwise (in which case the counter remains unchanged). After $n$ insertions, the value $n$ can be estimated by $\tilde{n} = (b^C-b)/(b-1) + 1$.

	\begin{definition}
		The \textit{approximate counting} problem is defined as follows. Each player $i$ is given a positive integer value $x_i \in \Z_{\geq 0}$, and the goal is for some player at the end to hold an estimate of $x = \sum_i x_i$.
	\end{definition}
	
	\begin{proposition}[Proposition 5 \cite{flajolet1985approximate}]\label{Prop:morris}
		If $C_n$ is the value of the Morris counter after $n$ updates, then $\ex{\tilde{n}  } =n$, and $\var{\tilde{n} } = (b-1)n(n+1)/2$.
	\end{proposition}
\begin{corollary}\label{Cor:Morris}
If $C_n$ is the value of a Morris counter run on a stream of $n$ insertions with base $b = (1 + (\eps \delta)^2)$, then with probability at least $1-\delta$, we have $\tilde{n} = (1 \pm \eps)n$ with probability at least $1-\delta$. Moreover, with probability at least $1-\delta$, the counter $C_n$ requires $O(\log \log(n) + \log(1/\eps) + \log(1/\delta))$-bits to store. 
\end{corollary}
\begin{proof}
	The first claim follows immediately from Proposition \ref{Prop:morris} and Chebyshev's inequality. Conditioned on  $\tilde{n} = (1 \pm \eps)n$, which occurs with probability at least $1-\delta$, we have that $\tilde{n} = (b^C-b)/(b-1) < 2n$, so $C = O( \frac{1}{1-b}\log(n)) = O(\frac{\log(n)}{(\eps \delta)^2})$, which can be stored in the stated space.
\end{proof}

 Next, we demonstrate that Morris counters can be \textit{merged}. That is, given one Morris counter $X$ run on a stream of length $n_1$ (of only insertions or deletions), and a Morris counter $Y$ run on a stream of length $n_2$, we can produce a Morris counter $Z$ which is distributed identically to a Morris counter run on a stream of $n_1 + n_2$ updates. 
 The procedure for doing so is given below.
	
	\begin{center}
		
		\fbox{\parbox{\textwidth}{ Merging Morris Counters: \\ \ttx{Input:} Morris counters $X,Y$ with base $b$. 
				\begin{itemize}
					\item 			
				Set $Z \leftarrow X$. 
				\item for $i = 1,2,\dots, Y$
					\[Z \leftarrow \begin{cases}
					Z + 1 & \text{with probability } b^{-Z + i - 1} \\
					Z & \text{ with probability } 
					\end{cases}\]
						\end{itemize}
					 \ttx{Output:} Morris counter $Z$ 
						}}
		
	\end{center}

A sketch of the proof of the correctness of the above merging procedure is given by Cohen \cite{CohenLecture13}. For completeness, we provide a proof here as well.	
	\begin{lemma}\label{lem:Merge}
		Given Morris counters $X,Y$ run on streams of length $n_1,n_2$ respectively, The above merging procedure produces a Morris counter $Z$ which is distributed identically to a Morris counter that was run on a stream of $n_1 + n_2$ insertions.
	\end{lemma}
\begin{proof}
Let $u_1,\dots,u_{n_1}$ be the updates that $X$ is run on, and let $v_{1},\dots, v_{n_2}$ be the updates that $Y$ is run on. Consider the distribution of a random variable $Z'$ run on $u_1,\dots,u_{n_1},v_1,\dots,v_{n_2}$. We show that $Z',Z$ have the same distribution.  We associate with each $u_i,v_i$ a uniform random variable $p_i ,q_i\in [0,1]$ respectively. First, fix $p_i$ exactly for all $i \in [n_1]$, and now condition on the set of $v_i$ which increment $Y$ during the execution. We now consider the condition distribution of the $q_i$'s given the updates $v_i$ which incremented $Y$. We show that the conditional distribution of executions of $Z'$ is the same as the merging procedure. 

To see this, conditioned on some counter value $y$, items $v_i$ that incremented $y$ has $q_i < b^{-y}$, otherwise they had $q_i \geq b^{-y}$. We will prove via induction that the items that: 1) the items $v_i$ that did not increment $Y$ do not increment $Z'$, 2) the item $v_j$ that had the $i$-th increment of $Y$, conditioned on the value of $Z'$ so far, increments $Z'$ with probability $b^{-Z + i -1}$, and 3) at any point during the processing of the $v_i$'s we have $Z' \geq Y$.

Note that at the beginning, we have $Z' \geq 0$ and $Y = 0$, so $Z' \geq Y$. Moreover, $v_1$ increments $Y$ with probability $1$, so conditioning on this does not affect the distribution of $q_1$. So the conditional probability that $Z'$ is incremented is just $b^{-Z'}$ as needed. If $Z' = 0$ initially, then $b^{-Z'} = 1$, so we maintain $Z' \geq 1 = Y$ at this point.

 Now consider any update $v_j$ that did not cause an increment to $Y$. We must have had $q_j \in [b^{-Y}, 1]$, and since we maintained $Z' \geq Y$ by induction, it follows that $v_j$ did not cause an update to  $Y$, and neither $Z'$ or $Y$ was affected by the update $v_j$, so we maintain $Z' \geq Y$. Now consider the item $v_j$ 
 that caused the $i$-th increment to $Y$. This conditioning implies $q_j \in [0,b^{-(i-1)}]$.  
 Conditioned on the value of $Z'$ so far, the probability $v_j$ increments $Z'$ is the probability that  $q_j \in [0,b^{-Z'}]$ conditioned on  $q_j \in [0,b^{-(i-1)}]$, which is $b^{-Z' + i - 1}$ as desired. Moreover, if we had that $Z' = Y$ at this time step, then $Z' = Y = (i-1)$, so $b^{-Z' + i - 1} = 1$ and $Z'$ is incremented with probability $1$ conditioned on the fact that $Y$ was incremented. So we maintain the property that $Z' \geq Y$ at all steps, which completes the induction proof. 	
\end{proof}

	\begin{corollary}\label{cor:l1}
	There is a protocol for $F_1$ estimation of non-negative vectors, equivalently for the approximate counting problem, in the message passing model which succeeds with probability $1-\delta$ and uses a max-communication of $O((\log\log(n) + \log(1/\eps) + \log(1/\delta))$-bits.
	\end{corollary}
\begin{proof}First note that the $F_1$ of $\X$ is simply $\sum_{i}\|X_i\|_1$, so each player can run a Morris counter on $\|X_i\|_1$ and send them up a tree to a center vertex $\mathcal{C}$ as in Section \ref{sec:randomround}.
	Since Morris counters can be merged without affecting the resulting distribution, the center player $\mathcal{C}$ at the end holds a Morris counter $Z$ which is distributed identically to a Morris counter with base $b$ run on $\|\X\|_1 = \sum_i\|X_i\|_1 $ updates. Then by Corollary \ref{Cor:Morris}, we have that $\tilde{n} = (b^Z - b)/(b-1) + 1 = (1 \pm \eps)\|\X\|_1$ with probability $1-\delta$, and moreover that $Z$ is at most $O(\log\log(n) + \log(1/\eps) + \log(1/\delta)$ bits long. Since $Z$ is at least as large as the size of any other counter sent across an edge, it follows that the max communication in the protocol was at most $O(\log\log(n) + \log(1/\eps) + \log(1/\delta)$ bits, and the proof of the total communication concludes from the fact that each player communicates exactly once. 
\end{proof}

We now note that Morris counters can easily used as approximate counters for streams with both insertions and deletions (positive and negative updates), by just storing a separate Morris counter for the insertions and deletions, and subtracting the estimate given by one from the other at the end. The error is then proportional to the total number of updates made.

\begin{corollary}\label{cor:Morriswithdeletions}
	Using two Morris counters separately for insertions and deletions, on a stream of $I$ insertions and $D$ deletions, there is an algorithm, called a signed Morris counter, which produces $\tilde{n}$ with $|\tilde{n} - n| \leq \eps(I+D)$, where $n = I-D$, with probability $1-\delta$, using space $O(\log\log(I+D) + \log(1/\eps) + \log(1/\delta ))$.
\end{corollary}

Hereafter, when we refer to a Morris counter that is run on a stream which contains both positive and negative updates as a \textit{signed} Morris counter. Therefore, the guarantee of Corollary \ref{cor:Morriswithdeletions} apply to such signed Morris counters, and moreover such signed Morris counters can be Merged as in Lemma \ref{lem:Merge} with the same guarantee, but just Merged separately the counters for insertions and deletions.


\subsection{The $F_p$ Estimation Protocol, $p <1$}

\begin{figure}[h!]
	\fbox{\parbox{\textwidth}{ Procedure for player $j$\\
			$k \leftarrow \Theta(1/\eps^2)$\\
			$\eps' \leftarrow \Theta(\eps\frac{\delta^{1/p}}{\log(n /\delta)} )$\\
			$\delta \leftarrow 1/(200k)$
			\begin{enumerate}[topsep=0pt,itemsep=-1ex,partopsep=1ex,parsep=1ex] 
				\item Using shared randomness, choose sketching matrix $S \in \R^{k \times n}$ of i.i.d. $p$-stable random variables, with $k = \Theta(1/\eps)$. Generate $S$ up to precision $\eta = \poly(1/(n,m,M))$, so that $\eta^{-1}S$ has integral entries. 
				\item For each $i \in [k]$, receive signed Morris counters $y_{j_1,i},y_{j_2,i},\dots,y_{j_t,i}$ from the $t \in \{0,\dots,m\}$ children of node $j$ in the prior layer. 
				\item Compute $\eta^{-1}\langle S_i, X_j \rangle \in \Z$, where $S_i$ is the $i$-th row of $S$, and run a new signed Morris counter $C$ on $\eta^{-1} \langle S_i, X_j \rangle$ with parameters $(\eps',\delta')$. 
				
				\item Merge the signed Morris counters $y_{j_1,i},y_{j_2,i},\dots,y_{j_t,i},C$ into a counter $y_{j,i}$.
				\item Send the merged signed Morris counter $y_{j,i}$ to the parent of player $j$. If player $j$ is the root node $\mathcal{C}$, then set $C_i$ to be the estimate of the signed Morris counter $y_{j,i}$, and return the estimate $$\eta \cdot \text{median} \; \{\frac{|C_1|}{\theta_p}, \dots, \frac{|C_k|}{\theta_p} \}$$
				Where $\theta_p$ is the median of the distribution $\mathcal{D}_p$.
				
			\end{enumerate}

	}}\caption{Multi-party $F_p$ estimation protocol, $p<1$} \label{fig:p<1}
\end{figure}

		We now provide our algorithm for $F_p$ estimation in the message passing model with $p \leq 1$. Our protocol is similar to our algorithm for $p \geq 1$. We fix a vertex $\mathcal{C}$ which is a center of the communication topology: i.e., the distance of the farthest vertex from $\mathcal{C}$ in $G$ is minimal among all such vertices. We then consider the shortest path tree $T$ rooted at $\mathcal{C}$, which has depth at most $d$, where $d$ is the diameter of $G$. The players then choose random vectors $S_i \in \R^n$ for $i \in [k]$, and the $j$-th player computes $\langle S_i , X_j \rangle$, and adds this value to a Morris counter. Each player receives Morris counters from their children in $T$, and thereafter merges these Morris counters with its own. Finally, it sends this merged Morris counter, containing updates from all players in the subtree rooted at $j$, to the parent of $j$ in $T$. At the end, the center $\mathcal{C}$ holds a Morris counter $C_i$ which approximates $\sum_j \langle S_i, X_j\rangle$. The main algorithm for each player $j$ is given formally in Figure \ref{fig:p<1}.

		We first prove a technical lemma which bounds the total number of updates made to the Morris counters. Since the error of the signed Morris counters is proportional to the number of updates made, and not to the actual value that is being approximated, it is important that the number of updates does not exceed the desired quantity of estimation by too much.

		\begin{lemma}\label{lem:L1bound}
			Let $Z_1,\dots,Z_n$ each be distributed via some distribution $\mathcal{D}$, but not necessarily independently, such that $\mathcal{D}$ has  the property that if $Z \sim D$ then there is a universal constant $c>0$ such that for every $\lambda >1$ we have $\pr{|Z| >c\lambda} < \lambda^{-p}$ for some $0 < p \leq 1$. Note that this is true of $D_p$.  Fix any positive vector $\X \in \R^n$, and let $Y_i = Z_i \X_i$ for $i \in [n]$. Then for any $\lambda \geq 1$, if $1-p > \theta > 0$ for some constant $\theta$, then
			\[ \bpr{\|Y\|_1 \geq C \lambda^{1/p} \|\X\|_p  } \leq  \frac{1}{\lambda}\]
			otherwise 
			\[ \bpr{\|Y\|_1 \geq C \lambda^{1/p} \log(n \lambda) \|\X\|_p  } \leq  \frac{1}{\lambda}\]
			Where $C>0$ is a fixed constant independent of $\lambda,n,p$.
		\end{lemma}
		\begin{proof}
			We proceed similarly as in Lemma \ref{Lemma:lpVariance}. For $k \in \Z$, let $I_k = \{ i \in [n]\; \big|\; 2^{-k} \|\X\|_p \leq  Y_i \leq 2^{-k+1}\|\X\|_p \}$. Note that for any $i \in [n]$
			\[ \bpr{i \in \bigcup_{k' \geq k} I_{k'}} \leq \bpr{Z_i \geq 2^{-k} \|\X\|_p/ \X_i}  \]
			\[ \leq c  2^{kp}  \frac{\X_i^p}{\|\X\|_p^p}		\]
			For some constant $c > 0$, where in the last inequality we used Proposition \ref{prop:stabletails} for $D_p$ (or the fact that $\mathcal{D}$ has tails of order $\lambda^{-p}$ by assumption). Thus $\ex{\sum_{k' \geq k}|I_{k'}|} \leq c 2^{pk}$. Let $\mathcal{E}_0$ be the event that $I_k = \emptyset$ for all $k \leq- \frac{1}{p}\log(10 c\lambda )$ for some sufficiently large constant $c$. By Markov's inequality, we have $\pr{\mathcal{E}_0 } \geq 1 - \frac{1}{10 \lambda}$. Now for any $k \geq \frac{1}{p} \log(n^2 /\lambda)$, the contribution of any coordinate of $Y_i$ to $\|Y\|_1$ with $i \in I_k$ is at most $\frac{\lambda }{n^2}\|\X\|_p$, thus the total contribution of all such items is at most $\frac{\lambda }{n}\|\X\|_p$, which we can safely disregard. Finally, the expected contribution of the coordinates $Y_i$ for $i \in I_k$ with $ - \lceil \frac{1}{p}\log(10 c\lambda )\rceil \leq  k \leq \lceil \frac{1}{p} \log(n^2 /\lambda) \rceil$ is at most 
			
			\[\bex{	\sum_{- \lceil \frac{1}{p}\log(10 c\lambda )\rceil \leq  k \leq \lceil \frac{1}{p} \log(n^2 /\lambda) \rceil} 2^{-k+1} \|\X\|_p|I_k|} \leq c \sum_{- \lceil \frac{1}{p}\log(10 c\lambda )\rceil \leq  k \leq \lceil \frac{1}{p} \log(n^2 /\lambda) \rceil}  2^{-k(1-p)+1 } \|\X\|_p \] 
			If $p$ is constant bounded from $1$, this sum is geometric and bounded above by $$O\left(2^{\log(10 c \lambda)(1/p - 1) +1 } \|\X\|_p\right) = O\left((10 c\lambda)^{1/p - 1} \|\X\|_p\right)$$ otherwise we can upper bound each term by this quantity, giving a total bound of $$O\left((10 c\lambda)^{1/p - 1} \frac{\log(n\lambda) }{p} \|\X\|_p\right) =O\left((10 c\lambda)^{1/p - 1} \log(n\lambda)  \|\X\|_p\right) $$ where the equality holds because $1/p = \Theta(1)$ if $p$ is close to $1$.  In the first case, by Markov's  inequality, we have
			
			\[\bpr{	\sum_{- \lceil \frac{1}{p}\log(10 c\lambda )\rceil \leq  k \leq \lceil \frac{1}{p} \log(n^2 /\lambda) \rceil} 2^{-k+1} \|\X\|_p|I_k| \geq 2(10 c \lambda)^{1/p} \|\X\|_p} \leq \frac{1}{2\lambda} \]
			Union bounding over the above event and $\mathcal{E}_0$, we obtain the desired result with probability at least $1-(\frac{1}{2\lambda} + \frac{1}{10 \lambda}) > 1-\frac{1}{\lambda}$ in the case that $p$ is constant bounded from $1$. In the event that it is not,  Markov's inequality gives:
			\[\bpr{	\sum_{- \lceil \frac{1}{p}\log(10 c\lambda )\rceil \leq  k \leq \lceil \frac{1}{p} \log(n^2 /\lambda) \rceil} 2^{-k+1} \|\X\|_p|I_k| \geq 2(10 c \lambda)^{1/p} \log(n \lambda) \|\X\|_p} \leq \frac{1}{2\lambda} \]
			as needed. Applying the same union bound gives the desired result for $p$ arbitrarily close or equal to $1$. 
			
		\end{proof}

\begin{theorem}\label{thm:morrismain}
	For $p \in (0,1)$, there is a protocol for $F_p$ estimation in the message passing model which succeeds with probability $2/3$ and uses a total communication of $O(\frac{m}{\eps^2}(\log\log(n) + \log(1/\eps))$-bits, and a max-communication of $O(\frac{1}{\eps^2}(\log\log(n) + \log(1/\eps))$-bits. The protocol requires a total of at most $d$ rounds, where $d$ is the diameter of the communication topology $G$.
\end{theorem}
\begin{proof}
	By Lemma \ref{lem:Merge}, the Merging procedure has the effect that at the end, the player $\mathcal{C}$ has a signed Morris counter $I,D \in \Z_{\geq 1}$, such that $I$ is distributed as a Morris counter run on the updates $\eta^{-1}\sum_{j : \langle S_i, X_j \rangle \geq 0} \langle S_i, X_j \rangle$, and $D$ is distributed as a Morris counter run on the updates $\eta^{-1}\sum_{j : \langle S_i, X_j \rangle < 0} \langle S_i, X_j \rangle$. Then By Corollary \ref{cor:Morriswithdeletions}, the estimate of this signed Morris counter is a value $\tilde{I} - \tilde{D}$ such that $|(\tilde{I} - \tilde{D}) - \eta^{-1}\sum_{j } \langle S_i, X_j \rangle| \leq \eps' \eta^{-1}\sum_{j } |\langle S_i, X_j \rangle| \leq\eps' \eta^{-1}\sum_{j=1}^n |S_{i,j}  \X_j|$,  with probability $1-\delta$. Call this event $E_{i}^1$. Now by Lemma \ref{lem:L1bound}, with probability at least $1-\delta$, we have that $\sum_{j=1}^n |S_{i,j}  \X_j| \leq C  \delta^{-1/p}\log(n \delta )\|\X\|_p$ for some constant $C$. Note that while Lemma \ref{lem:L1bound} held for random variables $S_{i,j}$ that were not generated to finite precision, note that generating up to precision $\eta$ only changes each term in the sum by at most $\eta M$, and since $\|\X\|_p \geq 1$ (because $\X \neq 0$, and if it was it could be tested with $O(1)$ bits of communication per player), this additive error can be absorbed into the constant $C$. So call the event that this inequality holds $E_{i}^2$, and let $E_i$ be the event that both $E_i^1$ and $E_i^2$ hold. Note $\pr{E_i} > 1-2 \delta$. We now condition on $\cap_{i=1}^k E_i$, which occurs with probability $1 - 2k \delta > 99/100$.

	Now, conditioned on  $\cap_{i=1}^k E_i$, it follows that for each $i \in [k]$, the center $\mathcal{C}$ has an estimate $C_i$ of $\eta^{-1}\langle S_i, \X \rangle$ such that
	\begin{equation}
	\begin{split}
	|\eta C_i - \langle S_i, \X \rangle| &\leq \eps'\sum_{j=1}^n |S_{i,j}  \X_j|\\
	& \leq C  \eps'\delta^{-1/p}\log(n /\delta )\|\X\|_p\\
	 &\leq \eps \|\X\|_p/4\\
	\end{split}
	\end{equation}
	Now by \cite{indyk2006stable}, setting $k = \Theta(1/\eps^2)$, if $S_j$ is the $j$-th row of $S$, then the median $\tilde{F}_p$ of the set $\big\{\frac{|\langle S_1, \X \rangle|}{\theta_p}, \dots, \frac{|\langle S_k, \X \rangle|}{\theta_p} \big\}$ satisfies $|\tilde{F}_p -  \|\X\|_p| < (\eps/2)  \|\X\|_p$ with probability at least $4/5$. Since $|\eta C_i - \langle S_i, \X \rangle| < (\eps/4)  \|\X\|_p$ for all $i \in [k]$, it follows that $\left| \eta \cdot \text{median} \; \{\frac{|C_1|}{\theta_p}, \dots, \frac{|C_k|}{\theta_p} \} - \|\X\|_p\right| \leq \eps \|\X\|_p$ with probability at least $1 - (1/100 + 1/5) > 3/4$, which completes the proof of correctness.
	
	For message complexity, note that conditioned on $\cap_{i=1}^k E_i$, every Morris counter in question is a $(1 + \eps)$ relative error approximation of the quantity (insertions or deletions) that it is estimating. Thus by Corollary \ref{cor:Morriswithdeletions}, noting that at most $\poly(n,m,M)$ updates are every made to any Morris counter, the space to store any Morris counter is $O(\log\log(n) + \log(1/\eps))$, which is an upper bound on the size of any message in the protocol. Note that we can safely fail if any message becomes larger than this threshold, since this would mean that $\cap_{i=1}^k E_i$ failed to hold. The total message complexity follows from the fact that each player sends at most one message during the entire protocol, and the round complexity follows from the depth of the shortest path tree rooted at the center $\mathcal{C}$.

\end{proof}
	
	\subsection{The Streaming Algorithm for $F_p$ Estimation, $p<1$}
	As discussed earlier, the insertion-only streaming model of computation is a special case of the above communication setting, where the graph in question is the line graph, and each player receives vector $X_i \in \R^n$ which is the standard basis vector $e_j \in \R^n$ for some $j \in [n]$. The only step remaining to fully generalize the result to the streaming setting is an adequate derandomization of the randomness required to generate the matrix $S$. Note that while it is a standard assumption in communication complexity that the players have access to an infinitely long shared random string, the likes of which is used to generate the matrix $S$ of $p$-stable random variables, in streaming, generally, the cost to store all randomness is counted against the space requirements of the algorithm. The model of streaming where this is not the case is known as the \textit{random-oracle model}.
	
	Our derandomization will follow from the results of $\cite{kane2010exact}$, which demonstrate that, using a slightly different estimator known as the log-cosine estimator (discussed below), the entries of each row $S_i$ can be generated with only $\Theta( \log(1/\eps)/\log\log(1/\eps))$-wise independence, and the seeds used to generate separate rows of $S_i$ need only be pairwise independent. Thus, storing the randomness used to generate $S$ requires only  $O(\frac{\log(1/\eps)}{\log\log(1/\eps)}\log(n))$-bits of space. 
	
	We now discuss the estimator of \cite{kane2010exact} precisely. The algorithm generates a matrix $S \in \R^{k \times n}$ and $S' \in \R^{k' \times n}$ with $k = \Theta(1/\eps^2)$ and $k' = \Theta(1)$, where each entry of $S,S'$ is drawn from $\mathcal{D}_p$. For a given row $i$ of $S$, the entries $S_{i,j}$ are $\Theta(\log(1/\eps)/\log\log(1/\eps))$-wise independent, and for $i \neq i'$, the seeds used to generate $\{S_{i,j}\}_{j=1}^n $ and $\{S_{i',j}\}_{j=1}^n$ are pairwise independent. $S'$ is generated with only $\Theta(1)$-wise independence between the entries in a given row in $S'$, and pairwise independence between rows. The algorithm then maintains the vectors $y = S\X$ and $y' = S'\X$ throughout the stream, where $\X \in \Z^n_{\geq 0}$ is the stream vector. Define $y_{med}' = \text{median}\{|y_i'| \}_{i=1}^{k'}/\theta_p$, where $\theta_p$ is the median of the distribution $\mathcal{D}_p$ (\cite{kane2010exact} discusses how this can be approximated to $(1 \pm \eps)$ efficiently). The log-cosine estimator $R$ of $\|\X\|_p$ is then given by
	\[R = y_{med}' \cdot \left(	- \ln\left(  \frac{1}{k}\sum_{i=1}^k  \cos\left( \frac{y_i}{y_{med}'}\right)\right)	\right)  \]

	\begin{theorem}\label{thm:streamingMain}
		There is a streaming algorithm for insertion only $F_p$ estimation, $p \in (0,1)$, outputs a value $\tilde{R}$ such that with probability at least $2/3$, we have that \[ |\tilde{R} - \|\X\|_p| \leq \eps \|\X\|_p\] where $\X \in \R^n$ is the state of the stream vector at the end of the stream. The algorithm uses $O((\frac{1}{\eps^2}(\log\log(n) + \log(1/\eps)) + \frac{\log(1/\eps)}{\log\log(1/\eps)}\log(n) )$-bits of space.
	\end{theorem}
	\begin{proof}
		In \cite{kane2010exact}, they first condition on the fact that $y_{med}'$ is a constant factor approximation of $\|\X\|_p$. Namely, they condition on the event that $|y_{med}' - \|\X\|_p| \leq \|\X\|_p/10$, which occurs with probability $7/8$ for sufficiently large $k' = \Theta(1)$ (this is just the standard $p$-stable median estimator of Indyk \cite{indyk2006stable}). Note that $y'$ can be computed exactly by our algorithm using only $O(\log(n))$ bits of space. Using our protocol from Theorem \ref{thm:morrismain}, applied to the line graph on $m$ updates (corresponding to a streaming algorithm), it follows that we can approximate $y$ by a vector $\tilde{y}\in \R^k$ such that $|\tilde{y} - y|_\infty < \eps' \|\X\|_p$ with probability $99/100$ (here we have boosted the probability by a constant by running $\Theta(1)$ copies in parallel and taking the median for each coordinate), for some sufficiently small $\eps' = \Theta(\eps)$, and such that computing and storing $\tilde{y}$ requires only $O(k( \log \log(n) + \log(1/\eps)))$-bits of space. Since cosine has bounded derivatives everywhere, we have $| \frac{\tilde{y_i}}{y_{med}'} -   \frac{y_i}{y_{med}'}  | \leq \frac{\eps' \|\X\|_p}{y_{med}'} = O(\eps')$, so  $|\cos\left( \frac{\tilde{y_i}}{y_{med}'}\right) -  \cos\left( \frac{y_i}{y_{med}'}\right)| < O(\eps')$ for all $i \in [k]$. This gives 
		$$\left|\frac{1}{k}\sum_{i=1}^k  \cos\left( \frac{\tilde{y}_i}{y_{med}'}\right) - \frac{1}{k}\sum_{i=1}^k  \cos\left( \frac{y_i}{y_{med}'}\right)\right| = O(\eps')$$
		Furthermore, conditioned on the success of the estimator of \cite{kane2010exact} (which occurs with probability $3/4$, and includes the conditioning on the event that $y_{med}' = \Theta(\|\X\|_p)$), we have $|\frac{1}{k}\sum_{i=1}^k  \cos\left( \frac{y_i}{y_{med}'}\right)| = \Theta(1)$ (see Lemma 2.6 in \cite{kane2010exact}), and since $\ln(\cdot)$ has bounded derivatives for values $\Theta(1)$, it follows that if 
			\[\tilde{R} = y_{med}' \cdot \left(	- \ln\left(  \frac{1}{k}\sum_{i=1}^k  \cos\left( \frac{\tilde{y}_i}{y_{med}'}\right)\right)	\right)  \]
		then $|\tilde{R} - R| < O(\eps'y_{med}' ) = O(\eps' \|\X\|_p)$. Taking $k = \Theta(1/\eps^2)$ with a small enough constant, we have $|R - \|\X\|_p| < \eps \|\X\|_p/2$, and setting $\eps' = \Theta(\eps)$ small enough, we obtain  $|\tilde{R} - \|\X\|_p| < \eps \|\X\|_p$ as needed.

	\end{proof}

	\section{Entropy Estimation} \label{sec:entropy}
	In this section, we show how our results imply improved algorithms for entropy estimation in the message-passing model. Here, for a vector $\X \in \R^n$, the Shannon entropy is given by $H = \sum_{i=1}^n \frac{|\X_i|}{\|\X\|_1} \log(  \frac{\||\X\|_1}{|\X_i|} )$. We follow the approach taken by \cite{clifford2013simple, li2011new,harvey2008sketching, harvey2008streamingB} for entropy estimation in data streams, which is to use sketched of independent \textit{maximally-skewed stable random variables}. While we introduced $p$-stable random variables in Definition \ref{def:stable} as the distribution with characteristic function $\ex{e^{itZ}	} = e^{-|t|^p}$, we remark now that the $p$-stable distribution is also parameterized by an additional \textit{skewness} parameter $\beta \in [-1,1]$. Up until this point, we have assumed $\beta = 0$. In this section, however, we will be using maximally skewed, meaning $\beta = -1$, $p=1$-stable random variables. We introduce these now

	\begin{definition}[Stable distribution, general]
		There is a distribution $F(p,\beta,\gamma,\delta)$ called the $p$-stable distribution with \textit{skewness} parameter $\beta \in [-1,1]$, scale $\gamma$, and position $\delta$. The characteristic function of a $Z \sim F(p,\beta,\gamma,\delta)$ variable $Z$ is given by:
		\[  \ex{e^{-i t Z}} = \begin{cases}
		\exp\left(-\gamma^p|t|^p \left[1- i \beta \tan(\frac{\pi p}{2}) \text{sign}(t) \right] + i \delta t	\right) & \text{ if } p \in (0,2] \setminus \{1\} \\ 
		\exp\left(- \gamma |t| \left[1+ i \beta \frac{2}{\pi}\text{sign}(t)\log(|t|) \right] + i \delta t	\right) & \text{ if } p =1 \\ 
		\end{cases}\]
		where $\text{sign}(t) \in \{1,-1\}$ is the sign of a real $t \in \R$. Moreover, if $Z \sim F(p,\beta,\gamma,0)$ for any $\beta \in [-1,1]$ and $0<p<2$, for any $\lambda > 0$ we have
		\[ \bpr{|Z| > C \lambda} \leq (\frac{\gamma}{\lambda})^p		\]
		where $C$ is some universal constant. 	
		We refer the reader to \cite{nolan2009stable} for a further discussion on the parameterization and behavior of $p$-stable distributions with varying rates.

	\end{definition}

	\begin{figure}[h!]
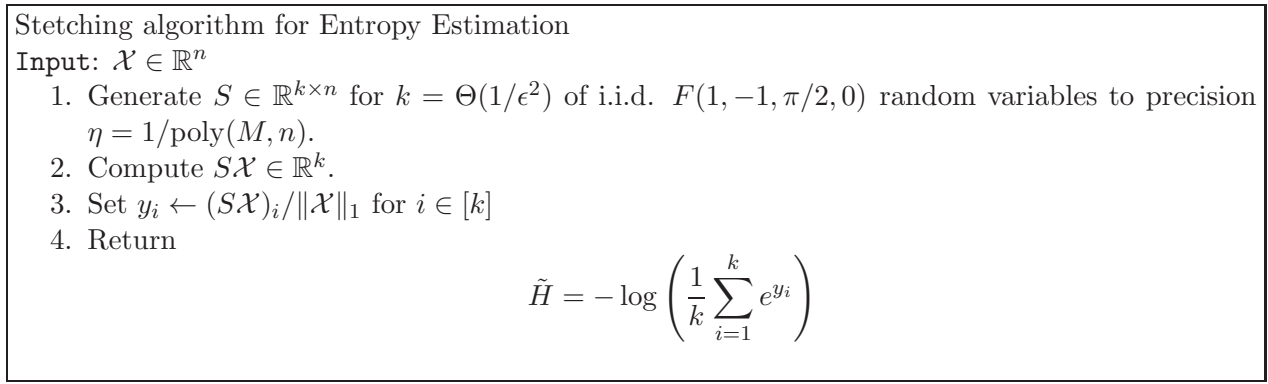

		\fbox{\parbox{\textwidth}{Stetching algorithm for Entropy Estimation\\
				\ttx{Input}: $\X \in \R^n$
				\begin{enumerate}[topsep=0pt,itemsep=-1ex,partopsep=1ex,parsep=1ex] 
					\item Generate $S \in \R^{k \times n}$ for $k = \Theta(1/\eps^2)$ of i.i.d. $F(1,-1,\pi/2,0)$ random variables to precision $\eta = 1/\poly(M,n)$. 
					\item Compute $S\X \in \R^k$. 
					\item Set $y_i \leftarrow (S\X)_i/\|\X\|_1$ for $i \in [k]$
					\item Return \[	\tilde{H} = - \log\left(\frac{1}{k}\sum_{i=1}^k e^{y_i}	\right)	\]
				\end{enumerate}
			
}}\caption{Entropy Estimation algorithm of \cite{clifford2013simple}} \label{fig:entropy}
\end{figure}

The algorithm of \cite{clifford2013simple} is given formally in Figure \ref{fig:entropy}. The guarantee of the algorithm is given in Theorem \ref{thm:cliff}. 
\begin{theorem}[\cite{clifford2013simple}]\label{thm:cliff}
	The above estimate $\tilde{H}$ satisfies $|\tilde{H} - H| < \eps$ with probability at least $9/10$.
\end{theorem}

\begin{lemma}\label{lem:entropy}
	Fix $0 < \eps_0 < \eps$. Let  $S \in \R^{k \times n}$ with $k = \Theta(1/\eps^2)$ be a matrix of i.i.d. $F(1,-1,\pi/2,0)$ random variables to precision $\eta = 1/\poly(M,n)$. Then there is a protocol in the message passing model that outputs $Y \in \R^k$ at a centralized vertex with $\|Y - S\X\|_\infty \leq \eps_0 \|\X\|_1$ with probability $9/10$. The protocol uses a total communication of $O(\frac{m}{\eps^2}(\log\log(n) + \log(1/\eps_0))$-bits, and a max-communication of $O(\frac{1}{\eps^2}(\log\log(n) + \log(1/\eps_0)))$-bits. 
\end{lemma}
\begin{proof}
	Note that in the proof of Theorem \ref{thm:morrismain}, the only property of the distribution of $S$ that was needed to obtain a $\eps_0 \|\X\|_1$ additive error estimate of $\langle S_i, \X\rangle = (S\X)_i$ for every  $i \in [k]$ was that Lemma \ref{lem:L1bound} hold for this distribution. Namely, we need the property that if $Z\sim F(1,-1,\pi/2,0)$, then there is a constant $c$ such that for every $\lambda \geq 1$ we have $\pr{|Z| \geq c \lambda} < \lambda^{-1}$. This tail bound is indeed true for $1$-stable distributions of any skewness (see, e.g. \cite{nolan2009stable}). Thus the conditions of Lemma \ref{lem:L1bound} hold for $F(1,-1,\pi/2,0)$, and the result follows from the proof of Theorem \ref{thm:morrismain}.
\end{proof}

\begin{theorem}\label{thm:entropymain}
There is a multi-party communication protocol in the message passing model that outputs a $\eps$-additive error of the Shannon entropy $H$. The protocol uses a max-communication of $O(\frac{1}{\eps^2}(\log\log(n) + \log(1/\eps))$-bits. 
\end{theorem}
\begin{proof}
	By Lemma \ref{lem:entropy}, the central vertex $\mathcal{C}$ (as in Section \ref{sec:randomround}) can obtain a vector $Y \in \R^k$ with $\|Y - S\X\|_\infty \leq \eps_0 \|\X\|_1$ with probability $9/10$. By Corollary \ref{cor:l1}, there is also a multi-party protocol which gives an estimate $R \in \R$ such that $R = (1\pm \eps_0)\|\X\|_1$ with probability $9/10$, where each player communicates once, sending at most $O(\log \log(n) + \log(1/\eps)))$ bits, and the estimate is held at $\mathcal{C}$ at the end. Now note that each $(S\X)_i/\|\X\|_1$ is distributed as $F(1,-1,\pi /2, H)$, where $H$ is the Shannon entropy of $\X$ (\cite{clifford2013simple} Lemma 2.2 or \cite{nolan2009stable} Proposition 1.17). By anti-concentration of all stable random variables, we have that if $Z \sim F(p,\beta,\gamma,\delta)$ then $\pr{|Z| < \eps_0 \gamma} < C \eps_0$ for some constant $C$. Thus $\pr{|(S\X)_i| < \eps_0  \|\X\|_1} \leq C \eps_0$ for some constant $C$, so we can condition on the event $\mathcal{E}$ that  $|(S\X)_i| > \eps_0/\eps^3  \|\X\|_1$ for all $i \in [k]$, which occurs with probability at least $99/100$ after a union bound and setting $\eps_0 = O( \eps^6)$ sufficiently small. Given this, we have $Y_i = (1 \pm \eps^3)(S\X)_i$ for all $i \in [k]$. Setting $\overline{Y_i} = Y_i / R$ for $i \in [k]$, we obtain $\overline{Y_i} = (1 \pm O(\eps^3))y_i$, where $y_i = (S\X)_i/\|\X\|_1$ as in Figure \ref{fig:entropy}. Moreover, we can condition on the event that $|(S\X)_i| < C'/(\eps^2)\|\X\|_1$ for all $i \in [k]$, which occurs with probability at least $99/100$ after a union bound and applying the tail bounds for $1$-stable random variables. Given this, we have $|y_i| =O( 1/\eps^2)$ for all $i \in [k]$, so $e^{\overline{Y}_i} = e^{y_i \pm O(\eps^3) y_i} = e^{y_i \pm O(\eps)} = (1\pm O(\eps))e^{y_i}$.
	Thus 
\begin{equation}
\begin{split}
 - \log\left(\frac{1}{k}\sum_{i=1}^k e^{\overline{Y_i}}	\right)	&  = - \log\left((1\pm O(\eps))\frac{1}{k}\sum_{i=1}^k e^{y_i}	\right) \\ 
   &= - \log\left(\frac{1}{k}\sum_{i=1}^k e^{y_i}	\right) + O(\eps)\\
    &= H + O(\eps)\\
\end{split}
\end{equation}
Where in the first equality, we use the fact that each summand $\sum_{i=1}^k e^{\overline{Y_i}}$ is non-negative, so $\sum_{i=1}^k ( 1 \pm O(\eps))e^{y_i} =  ( 1 \pm O(\eps)) \sum_{i=1}^k e^{y_i}$, and the last inequality follows from applying Theorem \ref{thm:cliff}. The central vertex $\mathcal{C}$ can then output $- \log\left(\frac{1}{k}\sum_{i=1}^k e^{\overline{Y_i}}	\right)$ as the estimate, which completes the proof.
\end{proof}

Since our protocol does not depend on the topology of $G$, a direct corollary is that we obtain a $\tilde{O}(\eps^{-2})$-bits of space \textit{streaming} algorithm for entropy estimation in the random oracle model. Recall that the random oracle model allows the streaming algorithm query access to an arbitrarily long tape of random bits. This fact is used to store the random sketching matrix $S$. 

\begin{theorem}\label{thm:entropystream}
There is a streaming algorithm for $\eps$-additive approximation of the empirical Shannon entropy of an insertion only stream in the random oracle model, which succeeds with probability $3/4$. The space required by the algorithm is $O(\frac{1}{\eps^2}(\log\log(n) + \log(1/\eps))$ bits.
\end{theorem}

	\section{Approximate Matrix Product in the Message Passing Model}\label{sec:amp}
		In this section, we consider the approximate regression problem in the message passing model over a topology $G= (V,E)$. Here, instead of vector valued inputs, each player is given as input two integral matrices $X_i \in \{0,1,2,\dots,M\}^{n \times t_1}$, $Y_i \in \{0,1,2,\dots,M\}^{n \times t_2}$. It is generally assumed that $n>> t_1,t_2$, so the matrices $X_i,Y_i$ are rectangular.	
	Let $\X = \sum_{i=1}^m  X_i$ and $\Y = \sum_i Y_i$. The goal of the players is to approximate the matrix product $ \X^T \Y \in \R^{t_1 \times t_2}$. Specifically, at the end of the protocol one player must output a matrix $R \in \R^{t_1 \times t_2}$ such that $\|R - \X^T \Y\|_F \leq \eps \|\X\|_F \|\Y\|_F $, where for a matrix $A$, $\|A\|_F = (\sum_{i,j} A_{i,j}^2)^{1/2}$ is the Frobenius norm of $A$. 
	
	
	We now describe a classic sketching algorithm which can be used to solve the approximate regression problem. The algorithm is relatively straightforward: it picks a random matrix $S \in \R^{k \times n}$. For instance, $S$ can be a matrix of i.i.d. Gaussian variables with variance $1/k$, or even the count-sketch matrix from Section \ref{sec:HH}. It then computes $S\X$ and $S\Y$, and outputs $(S\X)^T S\Y$. In this work, we will use a Gaussian sketch $S$. The following fact about such sketches will demonstrate correctness.

\begin{lemma}[\cite{kane2014sparser}]
		Fix matrices $\X \in \R^{n \times t_1}, \Y \in \R^{n \times t_2}$ and $0 < \eps_0$. Let $S \in \R^{k \times n}$ be a matrix of i.i.d. Gaussian random variables with variance $1/k$, for $k = \Theta(1/(\delta\eps_0^2))$. Then we have
		\[	\bpr{\|\X^T S^T S \Y - \X^T \Y \|_F \leq \eps_0 \|\X\|_F \|\Y\|_F } \geq1- \delta	\]
		Moreover, with the same probability we have $\|S\X\|_F = (1 \pm \eps_0) \|\X\|_F$ and $\|S\Y\|_F = (1 \pm \eps_0) \|\Y\|_F$
\end{lemma}
\begin{proof}
We first claim that dense Gaussian matrices $S \in \R^{k \times n}$ satisfy the $(\eps,\delta,2)$-JL moment property (Definition 20 \cite{kane2014sparser}) with $k = \Theta(1/(\delta \eps^2))$. To see this, note that $\|Sx\|_2^2 = \sum_{i=1}^k g_i^2 \|x\|_2^2/k$ by $2$-stability of Gaussians, where $g_i$'s are i.i.d. Gaussian. Thus $k\|Sx\|_2^2$ is the sum of squares of $k$ i.i.d. normal Gaussians, and therefore has variance $3k$ (since Gaussians have $4$-th moment equal to $3$). It follows that the variance of $\|Sx\|_2^2$, which is precisely $\ex{(\|Sx\|_2^2 - 1)^2}$, is $\Theta(1/k)$. Setting $k=1/(\delta \eps^2)$, it follows that $S$ has the $(\eps,\delta,2)$-JL moment property.  So by Theorem 21 of \cite{kane2014sparser}, we obtain the approximate matrix product result. 
\end{proof}
	
	Recall now by Lemma \ref{lemma:roundingmain}, with probability $1-\exp(-1/\delta)$ for a fixed $i,j$, taken over the randomness used to draw a row $S_i$ of $S$, we have that the central vertex $\mathcal{C}$ can recover a value $r_{\mathcal{C}}^{i,j}$ such that $\ex{r_{\mathcal{C}}^{i,j}} = (S\X)_{i,j}$ and $\var{ r_{\mathcal{C}}^{i,j}} \leq (\eps/\delta)^2 \|\X_{*,j} \|_2$, where $\X_{*,j}$ is the $j$-th \textit{column} of $\X$. Setting $\delta = \Omega(\log(n))$, we can union bound over this variance holding for all $i,j$. Recall the $\delta$ here just goes into the $\log(1/\delta)$ communication bound of the algorithm. Altogether, the central vertex obtains a random matrix $R^\X \in \R^{k \times t_1}$ such that $\ex{ R^\X} = (S \X)$ and $\ex{\|R^\X - S\X\|_F^2} \leq k (\eps/\delta)^2\sum_{j=1}^{t_1} \|\X_{*,j} \|_2$. Setting $\eps = \poly(1/k) = \poly(1/\eps_0)$ small enough, we obtain $\ex{\|R^\X - S\X\|_F^2}  \leq (\eps_0/\delta)^2 \|\X\|_F$. Similarly, we can obtain a $R^\Y$ at the central vertex $\mathcal{C}$, such that  $\ex{\|R^\Y - S\Y\|_F^2}  \leq (\eps_0/\delta)^2 \|\Y\|_F$. Let $\Delta^\X = R^\X - S\X$ and $\Delta^\Y = R^\Y - S\Y$. By Chebyshev's inequality, we have $\|\Delta^\X\|_F \leq \eps_0^2 \|\X\|_F$ and $\|\Delta^\Y\|_F \leq \eps_0^2 \|\Y\|_F$ with probability $1-\delta$, so
	
	\begin{equation}
	\begin{split}
	\| (R^\X)^T R^\Y - \X^T \Y\|_F &= \| \X^T S^T S \Y  + (\Delta^\X)^TS\Y  + \X^T S^T \Delta^\Y- \X^T \Y\|_F \\
&\leq \|\X^T S^T S \Y - \X^T \Y \|_F +\|\X^T S^T \Delta^\Y\|_F +\|(\Delta^\X)^TS\Y \|_F  \\	
&\leq\eps_0 \|\X\|_F\|\Y\|_F  +\|\X^T S^T\|_F \|\Delta^\Y\|_F +\|(\Delta^\X)^T\|_F\|S\Y \|_F  \\	
&\leq\eps_0 \|\X\|_F\|\Y\|_F + 2\eps_0^2 \|\Y\|_F \|\X\|_F   \\
& \leq O(\eps_0) \|\X\|_F\|\Y\|_F
	\end{split}
	\end{equation}
	where we applied Cauchy-Schwartz in the second inequality, giving the desired result. 	
	 Taken together, this gives the following theorem.

	\begin{theorem}\label{thm:approxmatrix}
		Given inputs $\X = \sum_{i=1}^m X_i, \Y = \sum_{i=1}^m Y_i$ as described above, there is a protocol which outputs, at the central vertex $\mathcal{C}$, a matrix $R \in \R^{t_1 \times t_2}$ such that with probability $3/4$ we have
		\[ 	\|R - \X^T \Y\|_F \leq \eps \|\X\|_F \|\Y\|_F 	\]
		The max communication required by the protocol is $O\left(\eps^{-2}(t_1 + t_2)( \log \log n + \log 1/\eps + \log d	)\right)$, where $d$ is the diameter of the communication topology $G$.
	\end{theorem}
\begin{proof}
	Correctness follows from the above discussion. The communication complexity bound is nearly identical to the proof of Theorem \ref{thm:LPmain}. First note that for each coordinate of $S\X$, exactly one message is sent by each player. Set $K =  (Mnm)^2/\gamma$, where $\gamma = (\eps \delta /(d \log(nm)))^C$ is the rounding parameter as in Lemma \ref{lemma:roundingmain}.  We can condition on the fact that $|S_{i,j}| \leq c K^3$ for all $i,j$ and for some constant $c>0$, which occurs with probability at least $1-1/K^2$. Now by Lemma \ref{lemma:roundingmain}, using the notation of Section \ref{sec:randomround}, for a fixed coordinate of $S\X$, we have that $\ex{e_i^2} \leq (j+1) \gamma_0^2 \sum_{v \in T_i} | \langle Q_v , Z \rangle|^2 \leq K^5$, where $e_i = \sum_{u \in T_i} X_u - r_i$, where $r_i$ is the message sent by the $i$-th player that coordinate of $SX_i$. By Markov's inequality with probability $1-1/K^2$ we have $|e_i| < K^4$. Thus $|r_i| \leq K^6$ for all $i$.

	Now for any $r_i$ for player $i$ in layer $\ell$ with $|r_i| < 1/(mK)^{d+3-\ell}$, we can simply send $0$ instead of $r_i$. Taken over all the potential children of a player $j$ in layer $\ell+1$, this introduces a total additive error of $1/K^{d+3-\ell}$ in $x_j$. Now if $x_j$ originally satisfies $|x_j| > 1/K^{d+2-\ell}$, then the probability that this  additive error of $1/K^{d+3-\ell}$ changes the rounding result $r_j = \Gamma(x_j)$ is $O(1/K)$, and we can union bound over all $m$ vertices that this never occurs, and then union bound over all $t_1k< n^2$ coordinates of $S\X$. Thus, the resulting $r_j$ is unchanged even though $x_j$ incurs additive error. Otherwise, if  $|x_j| \leq 1/K^{d+2-\ell}$, then since Player $j$ is in layer $(\ell+1)$, we round their sketch $x_j$ down to $0$ anyway. The final result is an additive error of at most $1/K^2$ to $r_{\mathcal{C}}$. Note that we can tolerate this error, as it is only greater than $\gamma\|\X\|_p$ when $\X = 0$, which is a case that can be detected with $O(1)$ bits of communication per player (just forwarding whether their input is equal to $0$). With these changes, it follows that $ 1/(mK)^{d+3} \leq r_j \leq K^6$ for all players $j$. Thus each message $r_j$ can be sent in $O( \log ( \log((mK)^{d+3} ) )) = O(\log \log(n) + \log(d) + \log(1/\eps))$ as needed, and the same bound holds for estimating $S\Y$.

\end{proof}

	\bibliography{cluster}
	
	\appendix 
    \section{Proof Sketch of the $\Omega(m/\eps^2)$ Lower Bound for $F_p$ estimation in the One-Way Coordinator Model}\label{app:1}
	We now sketch the proof of the  $\Omega(m/\eps^2)$ lower bound that was remarked upon in the introduction. First, consider the following problem Alice is given a vector $x \in \R^t$, and bob $y \in \R^t$, such that $x_i \geq 0, y_i \geq 0$ for all $i \in [t]$. Alice and Bob both send a message to Eve, who must then output a $(1 \pm \eps)$ approximation to $\|x + y\|_p$, for $p \in (0,2] \setminus \{1\}$. Via a reduction from the Gap-Hamming communication problem, there is an  $\Omega(1/\eps^2)$-bit communication lower bound for this problem \cite{woodruff2004optimal}. More specifically, there is a distribution $\mathcal{D}$ over inputs $(x,y) \in \R^t \times \R^t$, such that any communication protocol that solves the above problem on these inputs correctly with probability $3/4$ must send $\Omega(1/\eps^2)$ bits. 
	
	Now consider the one-way coordinator model, where there are $m$ players connected via an edge to a central coordinator. They are given inputs $x_1,\dots,x_m$, and must each send a single message to the coordinator, who then  must estimate $\|x\|_p = \|x_1 + x_2 + \dots + x_m\|_p$. Consider two distributions, $P_1,P_2$ over the inputs $(x_1,\dots,x_m)$. In the first, two players $i,j$ are chosen uniformly at random, and given as inputs $(x,y) \sim \mathcal{D}$, and the rest of the players are given the $0$ vector. In $P_2$, we draw $(x, y) \sim \mathcal{D}$, and every player is given either $x$ or $y$ at random. The players are then either given input from $P_1$ or $P_2$, with probability $1/2$ for each. In the first case, if the two players with the input do not send $\Omega(1/\eps^2)$ bits, then they will not be able to solve the estimation problem via the $2$-party lower bound. However, given only their input, the distributions $P_1$ and $P_2$ are indistinguishable to a given player. So the players cannot tell if the input is from $P_1$ or $P_2$, so any player that gets an non-zero input must assume they are in case $P_1$ if they want to solve the communication problem with sufficiently high constant probability, and send $\Omega(1/\eps^2)$ bits of communication. This results in $\Omega(m/\eps^2)$ total communication when the input is from $P_2$, which is the desired lower bound. 
	
		\section{$\Omega(1/\eps^2)$ Lower Bound for additive approximation of Entropy in Insertion-Only Streams}\label{app:2}
		We now prove the $\Omega(1/\eps^2)$-bits of space lower bound for any streaming algorithm that produces an approximation $\tilde{H}$ such that $|\tilde{H}  - H| < \eps$ with probability $3/4$. Here $H$ is the empirical entropy of the stream vector $\X$, namely $H = H(\X) = -\sum_{i=1}^n \frac{|\X_i|}{F_1} \log\frac{|\X_i|}{F_1} $. To prove the lower bound, we must first introduce the \textsc{GAP-HAMDIST} problem. Here, there are two players, Alice and Bob. Alice is given $x \in \{0,1\}^t$ and Bob receives $y \in \{0,1\}^t$. Let $\Delta(x,y) = |\{i \; | \; x_i \neq y_i\}|$ be the Hamming distance between two binary strings $x,y$. Bob is promised that either $\Delta(x,y) \leq t/2 - \sqrt{t}$ (NO instance) or $\Delta(x,y) \geq t/2 + \sqrt{t}$ (YES instance), and must decide which holds. Alice must send a single message to Bob, from which he must decide which case the inputs are in. It is known that any protocol which solves this problem with constant probability must send $\Omega(t)$-bits in the worst case (i.e., the maximum number of bits sent, taken over all inputs and random bits used by the protocol).
		
		\begin{proposition}[\cite{woodruff2004optimal, jayram2008one}]\label{prop:hamdist}
		    Any protocol which solves the \textsc{GAP-HAMDIST} problem with probability at least $2/3$ must send $\Omega(t)$-bits of communication in the worst case. 
		\end{proposition}
		
		We remark that while a $\Omega(1/\eps^2)$ lower bound is known for \textit{multiplicative-approximation} of the entropy, to the best of our knowledge there is no similar lower bound written in the literature for additive approximation in the insertion only model. We note that for the turnstile model, a lower bound of $\tilde{\Omega}(1/\eps^2 \log(n))$ for additive approximation is given in \cite{kane2010exact}. The proof of the following theorem is an adaptation of the proof in \cite{kane2010exact}, where we restrict the hard instance to have no deletions. 
		
		\begin{theorem}
		Any algorithm for $\eps$-additive approximation of the entropy $H$ of a stream, in the insertion-only model, which succeeds with probability at least $2/3$, requires space $\Omega(\eps^{-2})$
		\end{theorem}
		\begin{proof}
	 Given a $x,y \in \{0,1\}^t$  instance of  \textsc{GAP-HAMDIST}, for $t = \Theta(1/\eps^2)$, Alice constructs a stream on $2t$ items. Let $x'$ be the result of flipping all the bits of $x$, and let $x'' = x \circ 0^t + 0^t \circ x' \in \{0,1\}^{2t}$ where $\circ$ denotes concatenation. Define $y',y''$ similarly. Alice then inserts updates so that the stream vector $\X = x''$, and then sends the state of the streaming algorithm to Bob, who inserts his vector, so that now $\X = x'' + y''$. We demonstrate that the entropy of $H$ differs by an additive term of at least $\eps$ between the two cases.  In all cases case, we have 
		  \begin{equation}
		      \begin{split}
		           H & =  \frac{t-\Delta}{t} \log(t) + \frac{\Delta}{2t}\log(2t)\\
		           & = \log(t) + \Delta\left (\frac{2 \log(t)- \log 2t }{2t} \right)\\		    
		           \end{split}
		  \end{equation}

		   We can assume $t \geq 4$,  and then $2 \log(t)- \log (2t) = C > 0$, where $C$ is some fixed value known to both players that is bounded away from $0$. So as $\Delta$ increases, the entropy increases. Thus in a YES instance, the entropy is at least 
		    \begin{equation}
		      \begin{split}
		          H &\geq\log(t) + (t/2 + \sqrt{t})\frac{C}{2t} \\
		             &= \log(t) + (1/4 + 1/2 \sqrt{t})C\\ 	   
		             &= \log(t) + C/4 + \Theta(\eps) \\ 	  
		           \end{split}
		  \end{equation}
		  In addition, in the NO instance, the entropy is maximized when $\Delta = t/2 - \sqrt{T}$. so we have 
		     \begin{equation}
		      \begin{split}
		          H &\leq\log(t) + (t/2 - \sqrt{t})\frac{C}{2t} \\
		             &= \log(t) + C/4 - \Theta(\eps) \\ 	  
		           \end{split}
		  \end{equation}
		    Therefore, the entropy differs between YES and NO instances by at least an additive $\Theta(\eps)$ term. After sufficient rescaling of $\eps$ by a constant, we obtain our $\Omega(t) = \Omega(1/\eps^2)$ lower bound for additive entropy estimation via the linear lower bound for \textsc{GAP-HAMDIST} from Proposition \ref{prop:hamdist}. 
		    \end{proof}

\end{document}